\RequirePackage{fix-cm}
\documentclass[runningheads]{llncs}
\usepackage[utf8]{inputenc}
\usepackage[T1]{fontenc} 
\usepackage{lmodern} 
\rmfamily 
\DeclareFontShape{T1}{lmr}{b}{sc}{<->ssub*cmr/bx/sc}{}
\DeclareFontShape{T1}{lmr}{bx}{sc}{<->ssub*cmr/bx/sc}{}

\usepackage[a4paper,includeheadfoot,margin=3.6cm]{geometry}

\usepackage{amsmath,amssymb,amsfonts}
\usepackage{graphicx}
\usepackage{textcomp}
\usepackage{xcolor}

\usepackage{comment}
\usepackage{booktabs} 

\usepackage[utf8]{inputenc}
\usepackage{tikz}
\usetikzlibrary{positioning,automata,arrows,shapes.multipart,calc,decorations.pathreplacing,patterns,matrix,chains,scopes}

\usepackage{enumerate}
\usepackage{paralist}

\usepackage{xspace}
\usepackage{listings}

\usepackage[hidelinks]{hyperref}

\usepackage{subcaption}

\usepackage{wrapfig}

\newcommand\ourproto{{{Tenderbake}}\xspace}

\newcommand{\set}[1]{\{#1\}}
\newcommand{\Nat}{\mathbb{N}}
\newcommand{\pNat}{\mathbb{N}^*}

\makeatletter
\lst@Key{countblanklines}{true}[t]%
    {\lstKV@SetIf{#1}\lst@ifcountblanklines}

\lst@AddToHook{OnEmptyLine}{%
    \lst@ifnumberblanklines\else%
       \lst@ifcountblanklines\else%
         \advance\c@lstnumber-\@ne\relax%
       \fi%
    \fi}
\makeatother

\lstdefinelanguage{pseudocode}{
  numbers = left,
  firstnumber = auto,
  numberblanklines = false,
  countblanklines = false,
  morekeywords = {if, then, else, while, for, do, break, with, fun, proc, return, true, false, and, or, not, goto},
  morekeywords = {let, in, match, with, noop, assert, range, len},
  morekeywords = {phase, broadcast, pullChain, sleep, upon, of, request}, 
  columns=fullflexible,
  sensitive=true,
  commentstyle = \itshape\color{blue},
  morecomment={[l]\#},
  morecomment={[l]//},
  mathescape=true,
  basicstyle=\scriptsize,
  identifierstyle={\rmfamily},
  escapechar = \&,
  literate={++}{\texttt{++}}2,
  literate={:=}{{$\leftarrow\,$}}2,
  literate={->}{{$\rightarrow\;$}}2,
  captionpos=b,
  xleftmargin=20pt
}

\newcommand{\ls}{\lstinline[basicstyle=\normalsize]}

\newcommand{\ProposeM}{\ensuremath{\mathsf{Propose}}\xspace}
\newcommand{\PreendorseM}{\ensuremath{\mathsf{Preendorse}}\xspace}
\newcommand{\EndorseM}{\ensuremath{\mathsf{Endorse}}\xspace}

\newcommand{\PreendorsementsM}{\ensuremath{\mathsf{Preendorsements}}\xspace}

\newcommand{\ProposePL}{\ensuremath{\mathsf{PROPOSE}}\xspace}
\newcommand{\PreendorsePL}{\ensuremath{\mathsf{PREENDORSE}}\xspace}
\newcommand{\EndorsePL}{\ensuremath{\mathsf{ENDORSE}}\xspace}
\newcommand{\ProposeP}{\ProposePL}
\newcommand{\PreendorseP}{\PreendorsePL}
\newcommand{\EndorseP}{\EndorsePL}

\newcommand{\Phases}{\mathit{Phases}}

\newcommand{\observer}[1]{#1\text{-}\mathsf{observer}}

\newcommand{\block}[2]{\langle #1;#2\rangle} 

\newcommand{\roundDur}{\Delta}
\newcommand{\phaseDur}{\Delta'}

\newcommand{\var}[1]{\mathit{#1}} 
\newcommand{\validBlock}{\var{endorsableValue}}
\newcommand{\validRound}{\var{endorsableRound}}
\newcommand{\validPreendorsements}{\var{preendorsementQC}}
\newcommand{\lockedBlock}{\var{lockedValue}}
\newcommand{\lockedRound}{\var{lockedRound}}
\newcommand{\blockchain}{\var{blockchain}}

\newcommand{\msgs}{\var{messages}}
\newcommand{\msg}{\var{msg}}
\newcommand{\pc}{\var{proposalOrCertificate}}
\newcommand{\headQC}{\var{headCertificate}}
\newcommand{\dt}{\var{offset}}
\newcommand{\decision}{\var{block}}
\newcommand{\decisionCert}{\var{blockCertificate}}
\newcommand{\decisionOpt}{\var{decisionOption}}
\newcommand{\level}{\var{level}}
\newcommand{\round}{\var{round}}
\newcommand{\phase}{\var{phase}}
\newcommand{\offset}{\var{offset}}
\newcommand{\levelOffset}{\var{levelOffset}}
\newcommand{\rc}{\var{rt}}
\newcommand{\cp}{\var{cp}}

\newcommand{\fn}[1]{\ensuremath{\mathsf{#1}}} 
\newcommand{\pullChain}{\fn{pullChain}}
\newcommand{\pset}{\fn{preendorsements}}
\newcommand{\vset}{\fn{endorsements}}
\newcommand{\getBlock}{\fn{newValue}}
\newcommand{\proposer}{\fn{proposer}}
\newcommand{\isValid}{\fn{isValidValue}}
\newcommand{\waitMsg}{\fn{handleEvents}}
\newcommand{\selectProcesses}{\fn{committee}}
\newcommand{\selectProcessesL}{\fn{committeeAtLevel}}
\newcommand{\validChain}{\fn{validChain}}
\newcommand{\betterChain}{\fn{betterHead}}
\newcommand{\updateEndorsable}{\fn{updateEndorsable}}
\newcommand{\decisionRound}{\fn{round}}
\newcommand{\handleConsensusMessage}{\fn{handleConsensusMessage}}
\newcommand{\getNextRoundAndPhase}{\fn{getNextPhase}}
\newcommand{\hash}{\fn{hash}}
\newcommand{\proposal}{\fn{proposal}}
\newcommand{\pendos}{\fn{preendos}}
\newcommand{\legitimate}{\fn{isLegitimateValue}}
\newcommand{\consistent}{\fn{isConsistentValue}}
\newcommand{\validMessage}{\fn{isValidMessage}}
\newcommand{\runInstance}{\fn{runConsensusInstance}}
\newcommand{\startInstance}{\fn{startConsensusInstance}}
\newcommand{\initInstance}{\fn{initConsensusInstance}}
\newcommand{\filterMsgs}{\fn{filterMessages}}
\newcommand{\stopHandler}{\fn{stopEventHandler}}
\newcommand{\roundFromTime}{\Delta_{\mathit{inv}}}
\newcommand{\getDecision}{\fn{getDecision}}

\newcommand{\NewChain}{\fn{NewChain}}
\newcommand{\NewMessage}{\fn{NewMessage}}

\newcommand{\mbot}{\bot}

\newif\ifdraft
\drafttrue  

\ifdraft 

\usepackage[normalem]{ulem}
\newcommand\redout{\bgroup\markoverwith
	{\textcolor{red}{\rule[0.5ex]{2pt}{0.8pt}}}\ULon}

\newcommand{\correction}[2]{\redout{#1} {\bf #2}}

\definecolor{g0}{HTML}{339900}
\definecolor{o0}{HTML}{FF6600}
\definecolor{b0}{HTML}{0033CC}
\newcommand\q[1]{\textcolor{g0}{Q: #1}}
\newcommand\obs[1]{\textcolor{o0}{Obs: #1}}
\newcommand\todo[1]{\textcolor{red}{TODO: #1}}
\newcommand\fix[1]{\textcolor{o0}{\footnote{\textcolor{o0}{#1}}}}
\newcommand\new[1]{\textcolor{blue}{NEW: #1}}
\else 

\newcommand{\correction}[2]{}
\newcommand\q[1]{}
\newcommand\obs[1]{}
\newcommand\todo[1]{}
\newcommand\fix[1]{}
\newcommand\new[1]{}
\fi

\makeatletter
\tikzset{join/.code=\tikzset{after node path={%
\ifx\tikzchainprevious\pgfutil@empty\else(\tikzchainprevious)%
edge[every join]#1(\tikzchaincurrent)\fi}}}
\makeatother
\tikzset{>=stealth',every on chain/.append style={join},
         every join/.style={-}}
\tikzstyle{labeled}=[execute at begin node=$\scriptstyle,
   execute at end node=$]

\begin{document}

\title{{\bf{\ourproto}}\ --- A Solution to Dynamic Repeated Consensus for Blockchains}

\author{L\u acr\u amioara Aștef\u anoaei\inst{1},
  Pierre Chambart\inst{1},
  Antonella Del Pozzo\inst{2},
  Thibault Rieutord\inst{2,3},
  Sara Tucci\inst{2}, and
  Eugen Z\u alinescu\inst{1}}

\authorrunning{L.~Aștef\u anoaei, P.~Chambart, A.~Del Pozzo,
  T.~Rieutord, S.~Tucci, E.~Z\u alinescu}

\institute{Nomadic Labs, Paris, France \and CEA, List, F-91120, Palaiseau, France \and Université Paris-Saclay, Gif-sur-Yvette, France}
\maketitle
\begin{abstract}
  First-generation blockchains provide probabilistic finality: a block
can be revoked, albeit the probability decreases as the block
``sinks'' deeper into the chain. Recent proposals revisited
committee-based BFT consensus to provide deterministic finality: as
soon as a block is validated, it is never revoked. A distinguishing
characteristic of these second-generation blockchains over classical
BFT protocols is that committees change over time as the participation
and the blockchain state evolve. In this paper, we push forward in
this direction by proposing a formalization of the Dynamic Repeated
Consensus problem and by providing generic procedures to solve it in
the context of blockchains.


Our approach is modular in that one can plug in different
synchronizers and single-shot consensus instances. To offer a complete
solution,
we provide a concrete instantiation, called \ourproto, and present a blockchain
synchronizer and a single-shot consensus algorithm, working in a Byzantine and
partially synchronous system model with eventually synchronous
clocks. In contrast to recent proposals, our methodology is driven by
the need to \emph{bound} the message buffers. This is essential in
preventing spamming and run-time memory errors. Moreover, \ourproto
processes can synchronize with each other \emph{without} exchanging
messages, leveraging instead the information stored in the blockchain.

\end{abstract}

\section{Introduction}

Besides raising public interest, blockchains have also recently gained
traction in the scientific community. The underlying technology
combines advances in several domains, most notably from distributed
computing, cryptography, and economics, in order to provide novel
solutions for achieving trust in decentralized and dynamic
environments.

Our work has been initially motivated by
Tezos~\cite{tezos-white-paper,tezos-tutorial}, a blockchain platform
that distinguishes itself through its self-amendment mechanism:
protocol changes are proposed and voted upon. This feature makes
Tezos especially appealing as a testbed for experimenting with
different consensus algorithms to understand their strengths and
suitability in the blockchain context. Tezos relies upon a consensus
mechanism build on top of a liquid proof-of-stake system, meaning that
block production and voting rights are given
to participants in proportion to their stake and that participants can
delegate their rights to other stake-holders. As Nakamoto
consensus~\cite{bitcoin,GarayKL15}, Tezos' current consensus
algorithm~\cite{EmmyPlus_blog_post} achieves only probabilistic
finality assuming an attacker with at most half of the total stake,
and relying on a synchrony assumption.

The initial goal of this work was to strengthen the resilience of
Tezos through a BFT consensus protocol to achieve deterministic
finality while relaxing the synchrony assumption. We had two general
requirements that we found were missing in the existing BFT consensus
protocols. First, for security reasons, message buffers need to be
bounded: assuming unbounded buffers may lead to memory errors, which
can be caused either accidentally or maliciously, through spamming
for instance. Second, as previously observed~\cite{TendermintCorrectness},
plugging a classical BFT consensus protocol in a blockchain setting
with a proof-of-stake boils down to solve a form of repeated
consensus~\cite{delporte2008finite}, where each consensus instance
(i) produces a block, i.e., the decided value, and (ii) runs among
a committee of processes which are selected based on their stake.
To be applicable to open blockchains, committees need to be dynamic
and change frequently. Allowing frequent committee changes is
fundamental in blockchains for mainly two reasons: (i) it is not
desirable to let a committee be responsible for producing blocks
for too long, for neither fairness nor security; (ii) the stake of
participants may change frequently.

\smallskip

\noindent\textbf{Dynamic Repeated Consensus.} Typically, repeated
consensus is solved with state machine replication (SMR)
implementations. We, instead, propose to use a novel formalism,
dynamic repeated consensus (DRC) to take into account that, in the
context of \emph{open} blockchains, participants in consensus
change. To this end, we propose that the selection of participants is
based upon information readily available in the blockchains.



To solve DRC, we follow the methodology initially presented in \cite{DLS88} and
revived more recently in \cite{hotstuff,lumiere,Keidar20}: we decouple the logic
for synchronizing the processes in consensus instances from the consensus logic
itself. Thus our solution uses two main generic ingredients: a synchronizer and a
single-shot consensus skeleton. Our approach is modular in that one can plug in
different synchronizers and single-shot consensus algorithms. Our solution works
in a partially synchronous model where the bound on the message delay is
\emph{unknown}, and the communication is \emph{lossy} before the global
stabilization time (GST). We note that losing messages is a consequence of
processes having bounded memory: if a message is received when the buffers are
full, then it is dropped.



\smallskip

\noindent\textbf{Blockchain-based Synchronizer.}  The need for and the
benefits of decoupling the synchronizer from the consensus logic have
already been pointed out in
\cite{hotstuff,lumiere,Keidar20,Gotsman20}. Indeed, such separation of
concerns allows reusability and simpler proofs. We continue this line
of work and propose a synchronizer for DRC which does not exchange
messages. Instead, it relies upon local clocks while leveraging
information already stored in the blockchain. Our solution allows
buffers to be bounded and guarantees that correct processes in the
synchronous period are always in the same round, except for
negligeable periods of time due to clock drifts. Thus processes can
discard all the messages not associated with their current or next
round. This is in contrast with existing solutions, which, in
principle, need to store messages for an unbounded number of rounds.


\smallskip

\noindent\textbf{Consensus algorithm.} To complete our DRC solution, \ourproto,
we also present a single-shot consensus algorithm. Single-shot \ourproto is
inspired by Tendermint~\cite{tendermintv2,dissecting}, in turn inspired by
PBFT~\cite{pbft} and DLS~\cite{DLS88}. We improve Tendermint in two aspects: i)
we remove the reliable broadcast requirement, and ii) we provide faster
termination. Tendermint terminates once processes synchronize in the same round
after GST, in the worst case, in $n$ rounds, where $n$ is the size of the
committee. Single-shot \ourproto terminates in $f+2$ rounds, where $f$ is the
upper bound on the number of Byzantine processes. \ourproto departs from its
closest relatives Tendermint and HotStuff \cite{hotstuff} in that it is driven
by a bounded-buffers design leveraging a synchronizer that paces protocol phases
on timeouts only. However, the price for this is that \ourproto is not
optimistic responsive as HotStuff, which makes progress at the speed of the
network and terminates in $f+1$ rounds, at the cost of an additional phase. As a
last difference, we note that, contrary to recent pipelined
algorithms~\cite{hotstuff,pala}, \ourproto lends itself better to open
blockchains. Pipelined algorithms focus more on performance, however pipelining
imposes restrictions on how much and how frequently committees can
change~\cite{pala}.

%
%
%
%


\smallskip
\noindent\textbf{Further related work.}
We are not aware of any existing approach providing a complete, generic DRC
formalization. However, several references exist for particular aspects which we
touch upon. For instance, repeated consensus with bounded buffers has been
studied in
\cite{delporte2008finite,shahmirzadi2009relaxed} but in system models
which assume crash failures only. Working solutions for implementing dynamic
committees are (mostly partially) documented in
\cite{algorand,cosmos,dfinity,ibft2,libraSMR,PS18,smrdyn}.
The differences with respect to the closest relatives of single-shot \ourproto
have been discussed above.

\smallskip

\noindent\textbf{Outline.}  The paper is organized as follows:
Section~\ref{sec:model} defines the system model;
Section~\ref{sec:drc} formalizes the DRC problem and proposes a
generic solution; Section~\ref{sec:sync} proposes a synchronizer
leveraging blockchain's immutability;
Sections~\ref{sec:gssc}\;-\;\ref{sec:algorithm} present the
single-shot consensus skeleton and respectively single-shot \ourproto,
its instantiation; Section~\ref{sec:cc} discusses message complexity
and gives an upper bound on the recovery time after GST;
Section~\ref{sec:conclusion} concludes.
Appendix~\ref{sec:correctness} contains the detailed correctness
proofs of \ourproto.




\section{System Model}\label{sec:model}

We consider a message-passing distributed system composed of a possibly infinite
set~$\Pi$ of processes. Processes have access to digital signing and hashing
algorithms.
We assume that cryptography is perfect: digital signatures cannot be forged, and
there are no hash collisions.
Each process has an associated public/private key pair for signing and processes
can be identified by their public keys.

\smallskip

\noindent\textbf{Execution model.} Processes repeatedly run consensus instances
to decide \emph{output} values. New output values are appended to a \emph{chain}
that processes maintain locally. Consensus instances run in \emph{phases}. The
execution of a phase consists in broadcasting some messages (possibly none),
retrieving messages, and updating the process state.  At the end of a phase a
correct process exits the current phase and starts the next phase.  We consider
that message sending and state updating are instantaneous, because their
execution times are negligible in comparison to message transmission delays.
This means that the duration of a phase is given by the amount of time dedicated
to message retrieval.

\smallskip

\noindent\textbf{Partial synchrony.} We assume a partially synchronous system,
where after some unknown time~$\tau$ (the global stabilization time, GST) the
system becomes synchronous and channels reliable, that is, there is a finite
\emph{unknown} bound~$\delta$ on the message transfer delay.
Before~$\tau$ the system is asynchronous and channels are lossy.
%

We assume that processes have access to local clocks and that after $\tau$ these
clocks are loosely synchronized: at any time after~$\tau$, the difference
between the real time and the local clock of a process is bounded by some
constant~$\rho$, which, as $\delta$, is a priori unknown.

We consider that message sending and state updating are instantaneous, because
their execution times are negligible in comparison to message transmission
delays.

\smallskip

\noindent\textbf{Fault model.} Processes can be \textit{correct} or
\textit{faulty}. Correct processes follow the protocol, while faulty ones
exhibit Byzantine behavior by arbitrarily deviating from the protocol.
%
%

\smallskip

\noindent\textbf{Communication primitives.} We assume the presence of two
communication primitives built on top of point-to-point channels, where
exchanged messages are authenticated. The first primitive is a best-effort
broadcast primitive used by processes participating in a consensus instance and
the second is a pull primitive which can be used by any process.

Broadcasting messages is done by invoking the primitive $\fn{broadcast}$. This primitive
provides the following guarantees: (i)~{integrity}, meaning that each message is
delivered at most once and only if some process previously broadcast it;
(ii)~{validity}, meaning that after~$\tau$ if a correct process broadcasts a
message~$m$ at time~$t$, then every correct process receives $m$ by
time $t+\delta$.
For simplicity, we assume that processes also send messages to themselves.
Processes are notified of the reception of a message with a $\NewMessage$
event.

The $\pullChain$ primitive is used by a process to retrieve output
values from other processes. This primitive guarantees that, if
invoked by a process~$p$ at some time~$t>\tau$, then~$p$ will
eventually receive all the output values of correct processes
before~$t$.
%
%
We note that the pull primitive can be implemented in such a way that
the caller does not need to pull all output values, but only the ones
that it misses. Furthermore, output values can be grouped and thus
received as a chain of values.
Processes are notified of the reception of a chain with a $\NewChain$
event.

%



\section{Dynamic Repeated Consensus}
\label{sec:drc}



\subsection{Problem definition}
\label{ssec:def}
%

Originally, repeated consensus was defined as an infinite sequence of consensus
instances executed by the \textit{same} set of processes, with processes having
to agree on an infinitely growing sequence of decision values
\cite{delporte2008finite}. Dynamic repeated consensus, instead, considers that
each consensus instance is executed by a potentially different set of $n$
processes where~$n$ is a parameter of the problem.
More precisely, given the $i$-th consensus instance, only~$n$ processes
$\Pi_i\subseteq \Pi $ participate in the consensus instance proposing values and
deciding a unique value~$v_i$. Processes in $\Pi-\Pi_i$ can only adopt
$v_i$. Therefore output values can be either directly decided or adopted. We
assume that every correct process agrees a priori on a value~$v_0$.

To know the committee, each process has access to a deterministic
selection function $\selectProcesses$ that returns a sequence of
processes based on previous output values.
More precisely, the committee $\Pi_i$ is given by
$\selectProcesses([v_0])$ for~$i \leq k$ and by
$\selectProcesses(\bar{v}_p[..(i-k)])$ for $i > k$, where $k > 0$ is a
problem parameter, $\bar{v}_p$ denotes the sequence of output values
of process~$p$, and $\bar{s}[..j]$ denotes the prefix of length~$j+1$ of
the sequence~$\bar{s}$.
Each process calls $\selectProcesses$ with its own decided values;
however since decided values are agreed upon, $\selectProcesses$
returns the same sequence when called by different correct
processes. We note that the sets~$\Pi_i$ are potentially unrelated to
each other, and any pair of subsequent committees may differ.
However, we assume that in each committee, less than a third of the
members are faulty. For convenience, we consider the worst case:
$n=3f+1$, and each committee contains exactly $f$ faulty processes.
%


Dynamic repeated consensus, as repeated consensus, needs to satisfy
three properties: agreement, validity, and progress. Agreement and
progress have the same formulation for both problems. However, validity needs to
reflect the dynamic aspect of committees. To this end, we define
validity employing two predicates.
The first one is $\legitimate$. When given as input a
value~$v_i$, $\legitimate(v_i)$ returns true if the value has been
proposed by a legitimate process, e.g.,~a process in~$\Pi_i$.
%
The second predicate is $\consistent$. When given as input two
consecutive output values $v_i$, $v_{i-1}$, $\consistent(v_i, v_{i-1})$
returns true if~$v_i$ is consistent with~$v_{i-1}$. This
predicate takes into account the fact that an output value depends on
the previous one, as commonly assumed in blockchains. For instance,
when output values are blocks containing transactions, a valid block
must include the identifier or hash of the previous block, and
transactions must not conflict with those already decided.
%
%
For conciseness, we define $\isValid(v_i,v_{i-1})$ as a predicate
that returns true if both $\legitimate(v_i)$ and
$\consistent(v_i,v_{i-1})$ return true for $i > 0$.
Note that the use of an application-defined predicate for stating
validity already appears in~\cite{TendermintCorrectness,redbelly17}.

An algorithm that solves the \textit{Dynamic Repeated Consensus}
problem must satisfy the following three properties:
\begin{itemize}
\item (\textbf{agreement}) At any time, if $\bar{v}_p$ and $\bar{v}_q$ are the
  sequences of output values of two correct processes~$p$ and~$q$,
  then $\bar{v}_p$ is a prefix of $\bar{v}_q$ or $\bar{v}_q$ is a
  prefix of $\bar{v}_p$.
\item (\textbf{validity}) At any time, if $\bar{v}_p$ is the sequence of output values of a
  correct process~$p$, then the predicate
  $\isValid$($\bar{v}_p[i]$,$\bar{v}_p[i-1]$) is satisfied for any $i>0$.
\item (\textbf{progress}) The sequence of output values of every correct process grows.
\end{itemize}
We use $\bar{s}[i]$ to denote the $(i+1)$-th element of the
sequence~$\bar{s}$.

\subsection{A DRC solution for blockchains}
\label{sec:drcbco}

\subsubsection{Preliminaries.}

A \emph{blockchain} is a sequence of linked blocks.
The \emph{head} of a blockchain is the last block in the sequence. The
block \emph{level} is its position in the sequence, with the first
block having level~$0$. We call this block \emph{genesis}.
A block has a \emph{header} and a \emph{content}. The content
typically consists of a sequence of transactions; it is
application-specific and therefore we do not model it further. The
block header includes the level of the block and the hash of the
previous block, among other fields detailed later.

In a nutshell, the proposed DRC algorithm works as follows. At each level, 
processes run a single-shot consensus algorithm to agree on a
tuple~$(u, h)$, where~$b$ is a block proposed to be appended to the
blockchain, $u$ is the content of~$b$ and~$h$ is the hash of the
predecessor of~$b$.
Therefore, we consider that the output values in $\bar{v}$ from the
DRC definition in Section~\ref{ssec:def} are the agreed upon
tuples~$(u,h)$.

Intuitively, the block content is what needs to be agreed upon at a
given level.
%
%
Thanks to block hashes, the agreement obtained during a single-shot
consensus instance implies agreement on the whole blockchain,
\emph{except for its head}, for which there might not yet be agreement
on the other fields of the header besides the predecessor hash. The
possible ``disagreement'' comes from processes taking a decision at
possibly different times and thus on different proposed blocks which,
however, share the same content. Agreement on the head is obtained
implicitly at the next level. For clarity, we refer to a block as
being \emph{committed} if it is not the head of the blockchain of an
honest process.




In order for processes to validate a chain independently of the
current consensus instance, a certificate is included in the block
header to justify the decision on the previous block. A
\emph{certificate} is a quorum of signatures which serves as a
justification that the content of the predecessor block was agreed
upon by the ``right'' committee.
To effectively check certificates, the public keys of committee members are
stored in the blockchain. For instance, if the blocks up to level~$\ell$ contain
the public keys of the committee for level~$\ell+k$, and the genesis block
contains the public keys of the committees for levels~$1$ to~$k$, then the
signatures in a certificate at a given level can be checked by processes that
are not more than $k$ blocks behind.

\subsubsection{A DRC algorithm.}
\label{sec:drcbc}

Fig.~\ref{alg:drc} presents the pseudocode of a \emph{generic}
procedure to solve DRC in the context of blockchains. It is generic in
that it can run with \emph{any}
single-shot consensus algorithm.

Before going into details, we first enumerate that the variables
maintained by any correct process~$p$. Namely, the state of~$p$
contains:
\begin{itemize}
\item $\blockchain_p$, its local copy of the blockchain;
\item $\ell_p$, the current level at which $p$ runs a consensus
  instance, which equals the length of its blockchain;
\item $h_p$, the hash of the head of $\blockchain_p$, that is, of the
  block at level $\ell_p-1$;
\item $\headQC_p$, the certificate which justifies the head of
  $\blockchain_p$.
\end{itemize}
%
%
In the pseudocode, all these state variables are considered global,
while variables local to a procedure are those that do not have a
subscript.

Next, we refine the answer to $\pullChain$ requests, in that we consider that
the $\pullChain$ primitive retrieves more than just output values.
Concretely, 
when a correct process~$p$ at
level~$\ell_p$ answers a $\pullChain$ request, it returns a tuple
$(\blockchain_p, \pc)$ where $\blockchain_p$ is its local chain and
$\pc$ is either:
\begin{itemize}
\item the block that~$p$ considers as the current proposal at level~$\ell_p$ or,
\item in absence of a proposal, $\headQC_p$. 
\end{itemize}
Here, by \emph{proposal} we mean a proposed block.
%



%
\begin{figure}[t]
\begin{minipage}{0.46\textwidth}
  \begin{lstlisting}[name=Alg,xleftmargin=-.1cm]
proc $\fn{runDRC}()$
  schedule $\fn{onTimeoutPull}()$ to be executed after $I$ $\label{l:sp}$
  $\fn{updateState}([\var{genesis}],\emptyset)$ $\label{line:updateState-first}$
  while true
    $\initInstance()$
    $(\var{chain},\var{certificate})$ := $\runInstance()$ $\label{l:ri}$
    $\fn{updateState}(\var{chain},\var{certificate})$ $\label{line:nextLevel}$


proc $\fn{updateState}(\var{chain}, \var{certificate})$
  # NB: tail of $\color{blue}{\blockchain_p}$ is a prefix of $\color{blue}{\var{chain}}$
  $\blockchain_p$ := $\var{chain}$
  $\headQC_p$ := $\var{certificate}$
  $\ell_p$ := $\fn{length}(\var{chain})$ $\label{line:updateLevel}$
  $h_p$ := $\hash(\blockchain_p[\ell_p-1])$
  \end{lstlisting}
\end{minipage}
\begin{minipage}{0.5\textwidth}
  \begin{lstlisting}[name=Alg]
proc $\fn{onTimeoutPull}()$
  pullChain
  schedule $\fn{onTimeoutPull}()$ to be executed after $I$

proc $\waitMsg()$
  while not $\stopHandler()$ do $\label{line:timer}$
    upon $\NewMessage(\var{msg})$
      $\handleConsensusMessage(\var{msg})$
    upon $\NewChain(\var{chain},\pc)$
      $\var{certificate} := \fn{getCertificate}(\pc)$
      if $\validChain(\var{chain}, \var{certificate})$ then $\label{l:valid}$
        if $\fn{length}(\var{chain}) > \ell_p$ then
          return $(\var{chain}, \var{certificate})$ $\label{l:ret}$
        else if $\fn{length}(\var{chain}) = \ell_p\; \land$
           $\betterChain(\var{chain},\pc)$ then
          $\fn{updateState}(\var{chain}, \var{certificate})$ $\label{line:betterChain}$
    \end{lstlisting}
\end{minipage}
  \caption{DRC entry point and auxiliary procedures.}
  \label{alg:drc}
\end{figure}

We now proceed to describing the entry point of the DRC algorithm, that is, the
procedure~$\fn{runDRC}$ in Fig.~\ref{alg:drc}. We note that processes need not
start DRC at the same time. When executing $\fn{runDRC}$, a process starts by
scheduling calls to $\pullChain$. Then, using $\fn{updateState}$, it initializes
its local variables, namely the state variables already presented and the
variables specific to the single-shot algorithm. We use the function $\hash$ to
compute the hash of some input. The function $\fn{length}$ returns the length of
an input sequence.

After updating its state, the process iteratively runs consensus instances and
once an instance has finished, it updates its state accordingly. Normally, a
consensus instance simply decides on a value, and the corresponding block is
appended to the blockchain. However, a process might also be behind other
processes which have already taken decisions for more than one level. In this
case, as soon as the process invokes the $\pullChain$ primitive, it retrieves
missed decisions and thus possibly more blocks are appended to the blockchain.


We note that in the presence of dynamic committees, it is not enough
that processes call $\pullChain$ punctually when they are
behind. Indeed, assume that a process~$p$ decides at level~$\ell$ but
the others are not aware of this and have not decided, because the
relevant messages were lost; also assume that~$p$ is no longer a
member of the committee at level $\ell+1$, consequently, it no longer
broadcasts messages and thus the other processes cannot progress. To
solve this, each process invokes $\pullChain$ regularly, every $I$
time units, where $I>0$ is some constant.
\begin{remark}
  In practice, given an estimate $\delta_{\mathit{max}}$ of the
  maximum message delay~$\delta$, $I$ should be chosen bigger than
  $\delta_{\mathit{max}}$, as it is not useful to invoke the pull
  primitive more often than once every $\delta$ time~units.
\end{remark}


During the execution of a consensus instance, processes continuously
handle events to update their state.
The event processing loop is implemented by the $\waitMsg$ procedure in Fig.~\ref{alg:drc}.
The termination of the event handler is controlled by the $\stopHandler$
procedure, which is specific to the single-shot consensus algorithm.
There are two kinds of events: message receipts, represented by the
$\NewMessage$ event, and chain receipts, represented by
the $\NewChain$ event.
Upon receiving a new message $\msg$, a process~$p$ dispatches
it to the consensus instance.
Upon the receipt of a new chain,~$p$ updates accordingly:
\begin{itemize}
\item if the new chain is longer, and is valid,~$p$ updates its state
  and starts a new consensus instance for a higher level; this is
  because the \ls{return} on line~\ref{l:ret} passes the control back
  to the $\fn{runDRC}$ procedure in line~\ref{l:ri};
\item if the new chain has the same length but a better head,~$p$ only updates
  its state while remaining at the same level. Note however that only the
  DRC-related state is update, the single-shot instance continues unabated.
  The procedure $\betterChain$ is specific to the single-shot consensus
  algorithm. We detail it in Section~\ref{sec:algorithm}. For the moment, we
  note that by means of $\betterChain$, all processes have the same reference
  point for synchronization.
\end{itemize}
In addition to the chain, the $\NewChain$ event includes~$\pc$,
which serves as a justification that the head's value has indeed been
agreed upon.

The role of $\validChain(\var{chain})$ is two-fold:
\begin{enumerate}
\item to check whether $\var{chain}$'s head and the certificate from
  $\pc$ match; for this to be possible, we assume access to a
  procedure $\fn{getCertificate}$ provided at the single-shot
  consensus level;
\item to check whether for any level~$\ell$ the
  predicate~$\isValid(\var{chain}[\ell], \var{chain}[\ell-1])$ is
  satisfied; this means that the hash field in the header of the block $\var{chain}[\ell]$ equals
  $\fn{hash}(\var{chain}[\ell-1])$ (so that the
  predicate~$\consistent$ is satisfied), and that the value in each
  block is proposed by the right committee (so that the
  predicate~$\legitimate$ is satisfied); for the latter to be
  possible, certificates are stored in blocks as single-shot consensus
  specific elements (see Section~\ref{s:msg}).
\end{enumerate}



\begin{remark}
  The validity check on line~\ref{l:valid} is brought forward for
  clarity. In practice, since this check is heavier, one would first
  do the lighter checks on the length, respectively on the head.
\end{remark}

The DRC solution we presented is generic, one can instantiate it by
providing implementations to $\initInstance$, $\startInstance$,
$\fn{getCertificate}$, $\betterChain$, and $\stopHandler$. We show how
to concretely implement them in Sections~\ref{sec:gssc}
and~\ref{ssec:map}.


\section{A synchronizer for blockchains}\label{sec:sync}

We describe a synchronizer for \emph{round}-based consensus algorithms.
Round-based consensus algorithms progress in rounds, where, at each round,
processes attempt to reach a decision, and if they fail, they advance to the
next round to make another attempt.

In the context of round-based consensus algorithms, one standard way to achieve
termination of a single consensus instance is to ensure that processes remain at
the same round for a sufficiently long period of time \cite{DLS88,pbft}.
The synchronizer we propose realizes this by leveraging the immutability of the
blockchain.
One feature of our synchronizer is that it does not exchange any message thus it
does not increase the communication complexity. Instead of exchanging messages,
it relies on rounds having the same duration for all processes.
We require that round durations are increasing and unbounded.
Concretely, the duration of a round $r>0$ is given by $\roundDur(r)$,
where $\roundDur$ is a function with domain~$\Nat\setminus\set{0}$
such that, for any duration~$d\in\Nat$, there is a round~$r$ with
$\roundDur(r)\geq d$.
Furthermore, we assume that round durations are larger than the clock skew, so
that rounds are not skipped in the synchrony period. 

\begin{remark}\label{remark:choosingDelta}
In practice, given estimates $\delta_{\mathit{real}}$ of the real
message delay and $\delta_{\mathit{max}}$ of the maximum message
delay~$\delta$, we would choose $\roundDur$ such that: (i)
$\roundDur(1)$ is slightly bigger than~$\delta_{\mathit{real}}$, (ii)
$\roundDur$ increases rapidly (e.g.~exponentially) till it
reaches~$\delta_{\mathit{max}}$, and (iii) then it increase slowly
(e.g.~linearly) afterwards.
\end{remark}

To determine at which round the process should be, the synchronizer
relies on local clocks. Therefore, when clocks are synchronized, all
processes will be at the same round. However, a prerequisite is that
processes agree on the starting time of the current instance.
As different processes may decide at different rounds, and therefore
at different times, there is a priori no consensus about the start
time of an instance.
We adopt a solution based on the following observation: if the round at which a
decision is taken is eventually known by all processes, then they can agree on a
common global round at which a consensus instance is considered to have
terminated. Indeed, a process considers that the consensus instance has ended at
the smallest round at which some process has decided.

The above solution can be implemented by (1) considering that a block
header stores the round at which the block is produced, and
(2) using the $\betterChain$ procedure, which is called by a process~$p$ at
line~\ref{line:betterChain} upon receiving a new chain in response to
a $\pullChain$ request. This procedure checks
if some other process has already taken a decision sooner, in terms of
rounds. If this is the case, $\betterChain$ signals to its caller that
it is ``behind'' and thus that it needs to resynchronize. We postpone
the concrete implementation of~$\betterChain$ to
Section~\ref{ssec:baker} because it is specific to the single-shot
consensus algorithm. For the moment, to illustrate the role of
$\betterChain$, Fig.~\ref{fig:updatechain} shows an update of the head
of a process $p$'s blockchain. Initially, the head of~$p$'s local
chain is~$b'$. Then, $p$ sees the block~$b''$ at level~$\ell$ with a
smaller round than~$b'$ and therefore updates the head of its local
chain to~$b''$.

\begin{figure}[t]
  \centering
  \begin{tikzpicture}
  [
    grow                    = right,
    sibling distance        = 2em,
    level distance          = 6em,
    edge from parent/.style = {draw, -latex},
    treenode/.style         = {shape=rectangle, draw, align=center, top color=white, rounded corners},
    env/.style      = {treenode}, 
    com/.style      = {treenode, minimum size = .9cm, thick, top color = black},
    dummy/.style    = {circle},
    sloped,
    minimum size=.2cm,
    thick,
    brace/.style={
        decoration={brace},
        decorate
    },
    bracem/.style={
        decoration={brace, mirror},
        decorate
    },
    font=\sffamily,
    scale=0.85,
    every node/.style={transform shape}
  ]
  \node [com, label=above:$0$] (root) {}
  child { node [com, label=above:$1$] {}
    child { node [dummy] {\hphantom{aa}.....\hphantom{aa}}
      child { node [com, label=above:$\ell_p-2$, label=below:$b$] (stop) {}
        child { node [dashed, env, label=below:$b''$] (stopq) {$(u,h)$\\$r=2$} edge from parent [<-]}
        child [missing]
        child [missing]
          child { node [env, label=above:$\ell_p-1$, label=below:$b'$] (stopp) {$(u,h)$\\$r=3$} edge from parent [<-]}
      edge from parent [<-] node [below] {} }
      edge from parent [<-] node [below] {} }
    edge from parent [<-] node [below] {} }
  ;

  \draw [bracem] ($(root.west) - (0,1)$) -- node [below=.03cm
    ,font=\fontsize{10}{0}\selectfont
  ] {committed blocks} ($(stop.east) - (0, 1)$);

  \draw [bracem]
  let \p1 = (root.west) in
  let \p2 = (stopq.east) in
  ($(root.west) - (0,2)$) --
  node [below=.03cm
    ,font=\fontsize{10}{0}\selectfont
  ] {$q$'s local chain}
  ($(stopq.east) - (0, 2) + (0, \y1-\y2)$);

  \draw [brace]
  let \p1 = (root.west) in
  let \p2 = (stopp.east) in
  ($(root.west) + (0,2.1)$) --
  node [above=.03cm
    ,font=\fontsize{10}{0}\selectfont
  ] {$p$'s local chain}
  ($(stopp.east) + (0,2.1) + (0, \y1-\y2)$);
  \end{tikzpicture}
  \caption{An update of the head of $p$'s blockchain. Solid boxes represent
    blocks in~$p$'s blockchain before the update, while the dashed box
    represents the block that triggers the update. Block levels and labels are
    given above and respectively below the corresponding boxes. The hash $h$ is
    that of block~$b$.
  }
  \label{fig:updatechain}
\end{figure}
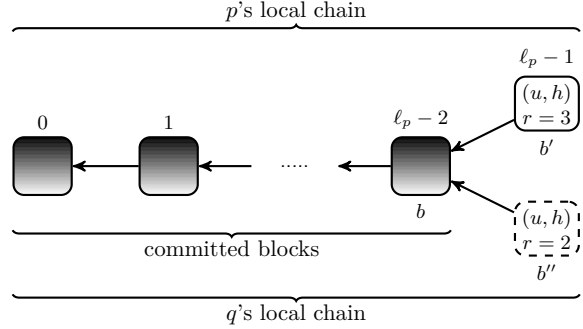

Finally, we present the synchronization procedure in Fig.~\ref{fig:sync}. We
assume that the genesis block contains the time~$t_0$ of its creation.
To synchronize, $p$ uses its local clock, whose value is obtained by calling
$\fn{now}()$, and the rounds of the blocks in its blockchain to find out what
its current round and the time position within this round should be. Process~$p$
determines first the starting time of the current level and stores it
in~$t$. To do this,~$p$ adds the durations of all rounds for all
previous levels.
 Once~$p$ has determined~$t$, it finds the current round by checking
 incrementally, starting from round $r=1$ whether the round $r$ is the current
 round: $r$ is the current round if there is no time left to execute a higher
 round. The variable $t$ is updated to represent the time at which round~$r$
 started.  The difference $t' - t$ represents the offset between the beginning
 of the round $r$ and the current time.

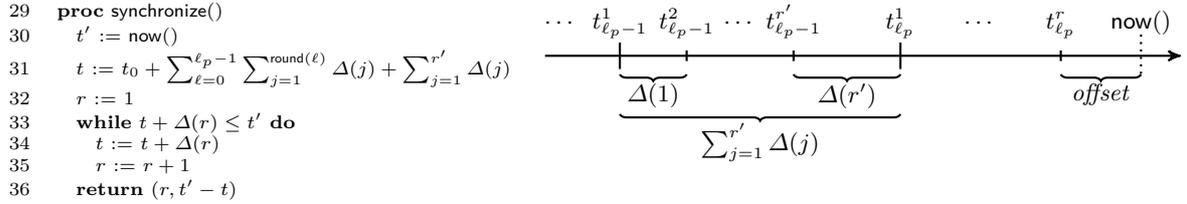
\begin{figure}[t]
\begin{minipage}{0.45\textwidth}
\begin{lstlisting}[name=Alg, xleftmargin=-.02cm]
proc $\fn{synchronize}()$
  $t'$ := $\fn{now}()$
  $t$ := $t_0+\sum_{\ell=0}^{\ell_p-1}\sum_{j=1}^{\decisionRound(\ell)}\roundDur(j) +\sum_{j=1}^{r'}\roundDur(j)$ $\label{line:tsum}$
  $r$ := $1$
  while $t + \roundDur(r) \leq t'$ do $\label{line:tloop}$
    $t$ := $t + \roundDur(r)$
    $r$ := $r + 1$
  return $(r, t'-t)$
\end{lstlisting}
\end{minipage}
\begin{minipage}{.6\textwidth}
  \begin{tikzpicture}[y=10pt, x=10pt,
      thick, font=\footnotesize]

    \tikzset{
      position label/.style={
        above = 3pt,
        text height = 1.5ex,
        text depth = 1ex
      },
      brace/.style={
        decoration={brace, mirror},
        decorate,
      }
    }

    \draw [->] (4.7,0) -- (28.5,0);

    \foreach \x in {7.5,18}
    \draw (\x,5pt) -- (\x,-5pt);

    \foreach \x in {10,14,24}
    \draw (\x,2pt) -- (\x,-2pt);


    \node [position label] (B) at (5.3,0) {$\cdots$};
    \node [position label] (B0') at (7.5,0) {$t_{\ell_p-1}^1$};
    \node [position label] (B1') at (10,0) {$t_{\ell_p-1}^2$};
    \node [position label] (dots0) at (12,0) {$\cdots$};
    \node [position label] (B2') at (14,0) {$t_{\ell_p-1}^{r'}$};
    \node [position label] (B2'') at (18,0) {$t_{\ell_p}^{1}$};
    \node [position label] (dots) at (21,0) {$\cdots$};
    \node [position label] (Be') at (24,0) {$t_{\ell_p}^{r}$};

    \draw [-, dotted] (27,8pt) -- (27,-8pt);
    \node [position label] (now) at (27,0) {$\fn{now}()$};
    \draw [brace] (24,-7pt) -- node [below] {$\var{offset}$} ($(now.south) - (0,1.)$);

    \draw [brace] ($(B0'.south) - (0,1.)$) -- node [below] {$\roundDur(1)$} ($(B1'.south) - (0,1.)$);

    \draw [brace] ($(B2'.south) - (0,1.)$) -- node [below] {$\roundDur(r')$} ($(B2''.south) - (0,1.)$);

    \draw [brace] ($(B0'.south) - (0,2.5)$) -- node [below] {$\sum_{j=1}^{r'}\roundDur(j)$} ($(B2''.south) - (0,2.5)$);



  \end{tikzpicture}
\end{minipage}
\caption{A round-based synchronizer and a timeline. Small/large vertical lines
  represent round/level boundaries, respectively.}
\label{fig:sync}
\end{figure}

\begin{remark}
  If blocks include their creation timestamp, then the sum at
  line~\ref{line:tsum} can be more efficiently computed by adding the
  timestamp of the head and the duration of the rounds for the current
  level.
  For simplicity, we omit timestamps in blocks.
\end{remark}

Fig.~\ref{fig:sync} also illustrates the timeline of a process that increments
its rounds using the procedure~\fn{synchronize}, where $t_{\ell_p}^r$ represents
the starting time of the round~$r$ of level~$\ell_p$ and $r'$ stands for
the last round of level~$\ell_p-1$. The figure also illustrates the offset $t' - t$.

\begin{remark}
  \label{r:sync}
  When the synchronizer is called is specific to the consensus instance. This
  should be done such that processes which are ``behind'' resynchronize in a
  timely manner. This means, in particular, that
  %
invoking the synchronizer should not be based only on the $\NewMessage$ and
  $\NewChain$ events: in case the event triggering the synchronization attempt
  occurs before $\tau$, the attempt may fail if clocks are not synchronized.
  One way to ensure that processes eventually synchronize is for the consensus
  instance, trigger the synchronization at regular time intervals.
\end{remark}

\section{A Single-Shot Consensus Skeleton}
\label{sec:gssc}

In this section we give a generic implementation for the procedure
$\runInstance$ from Section~\ref{sec:drcbc}.
Here we make another standard assumption on the structure of the single-shot
consensus algorithm, namely that each round evolves in sequential
\emph{phases}. For instance, PBFT in normal mode has 3 phases (named
\emph{pre-prepare}, \emph{prepare}, and \emph{commit}), Tendermint as well, DLS
and Hotstuff have 4 phases, etc.

We let $m$ denote the number of phases.
As for rounds, we assume that each phase has a predetermined duration. The duration
is given by the round~$r$ it belongs to, and it is denoted
$\phaseDur(r)$. For simplicity, we assume that $\roundDur(r) =
m\cdot\phaseDur(r)$.
We also refine the assumption on round durations, and also require that phase
durations are larger than the clock skew, so that phases are not skipped in the
synchrony period, i.e. $\Delta'(1)>2\rho$.

To synchronize correctly, a process also needs to update its phase (not only its
round) and to know its time position within a phase.
These can be readily determined from the round and the round offset
returned by \fn{synchronize}. The procedure $\getNextRoundAndPhase$,
presented in Fig.~\ref{alg:gss}, performs this task.
For the pseudocode, we consider that each phase has a label identifying it and
we use $\var{phases}$ to denote the sequence of phase labels.
%


\begin{figure}[t]
 \begin{minipage}{0.55\textwidth}
   \begin{lstlisting}[name=Alg,xleftmargin=-.05cm]
proc $\runInstance()$
  $(\var{round},\var{roundOffset})=\fn{synchronize}()$ $\label{line:sync}$
  if $r_p>\var{round}$ then $\label{line:isAhead}$ # $\textcolor{blue}{p}$ is ``ahead''
    $\runInstance()$ $\label{line:runCIreccall}$
  else  # $\textcolor{blue}{p}$ is ``behind''
    $(\var{phase},\var{phaseOffset})$ := $\getNextRoundAndPhase(\var{round},\var{roundOffset})$ $\label{line:getPhase}$
    $r_p$ := $\var{round}$ $\label{l:rs}$
    set $\var{runEventHandler}$ timer to $\phaseDur(r_p) - \var{phaseOffset}$ $\label{line:setT}$
    if $p \in \selectProcessesL(\ell_p)$ then $\label{line:checkValidatorWaitMsg}$
      goto $\var{phase}$ $\label{line:jumpToPhase}$
    else
      goto $\observer{\var{phase}}$ $\label{line:jumpToObserverPhase}$
  \end{lstlisting}
 \end{minipage}
 \begin{minipage}{0.4\textwidth}
   \begin{lstlisting}[name=Alg]
proc $\getNextRoundAndPhase(\var{round},\var{roundOffset})$
  $i$ := $\var{roundOffset}\; /\; \phaseDur(\var{round})$
  $\var{phase}$ := $\var{phases}[i]$
  $\var{phaseOffset}$ := $\var{roundOffset}- i\cdot\phaseDur(\var{round})$
  return $(\var{phase}, \var{phaseOffset})$

proc $\fn{advance}(\decisionOpt)$ $\label{l:advance}$
  match $\decisionOpt$ with
    | $\mathsf{Some}$ $(\decision, \decisionCert)$ ->
       $c$ := $\blockchain_p$ ++ $\decision$
       return $(c, \decisionCert)$ $\label{l:reta}$
    | $\mathsf{None}$ -> # no decision
       $r_p$ := $r_p + 1$  $\label{line:roundIncrement}$
       $\filterMsgs()$
       $\runInstance()$  $\label{line:runInstance}$
\end{lstlisting}
 \end{minipage}
  \caption{Entry point and progress procedures for generic single-shot consensus.}
  \label{alg:gss}
\end{figure}

The entry point of a single-shot consensus instance is $\runInstance$, given in
Fig.~\ref{alg:gss}.
As part of its state, a process~$p$ also maintains its current round~$r_p$.
A process~$p$ starts by calling \fn{synchronize} in an attempt to
(re)synchronize
with other processes. We recall that this is just an attempt and not a guarantee
because clocks are not necessarily synchronized before~$\tau$.
If \fn{synchronize} returns that $p$ should be at a round in the past with
respect to~$p$'s current round, then $p$ invokes (indirectly) the synchronizer
again. This active waiting loop ensures that $p$ is ready to continue its
execution as soon as it is not ``ahead'' anymore. We note that a jump backward
to a previous round or phase may jeopardize safety.
%
When $p$ is ``behind'',
%
%
it first uses the procedure
$\getNextRoundAndPhase$ to obtain the phase at which it should be.
%
%
%
Next, it updates its round and the timer used to time the execution of the event
handler. Concretely, through this timer, the generic procedure $\stopHandler$
(used by~$\waitMsg$ at line~\ref{line:timer} in Fig.~\ref{alg:drc}) is
implemented as follows:

\noindent
\begin{minipage}{\linewidth}
\begin{lstlisting}[name=Alg]
proc $\stopHandler()$
  return true iff timer $\var{runEventHandler}$ expired
\end{lstlisting}
\end{minipage}
Afterwards, $p$ checks whether it is part of the committee for level~$\ell_p$. 
%
To this end, we assume having access to a $\selectProcessesL$ function, which
returns the committee at some given level~$\ell$. This function corresponds to
$\selectProcesses(\bar{v}_p[..(\ell-k)])$ (Section~\ref{ssec:def}), where
$\bar{v}_p$ is the sequence of output values of the caller process~$p$.
Finally, $p$ executes the single-shot consensus algorithm according to its role and
to the phase returned by $\getNextRoundAndPhase$.
The determined phase is executed by means of an unconditional jump to
corresponding phase label.
The two \ls{goto} statements in Fig.~\ref{alg:gss} are intentionally symmetric
for committee and non-committee members to keep all processes in sync. Doing so
also for non-committee members has the advantage of not introducing delays when
they eventually become part of the committee.

Fig.~\ref{alg:gss} also shows the $\fn{advance}$ procedure, which is used by
processes to handle the progress of the current consensus instance by either
returning the control to $\fn{runDRC}$ when a decision can been taken at the
current round; or otherwise increasing the round. In this former case,
$\fn{advance}$ first prepares the updated blockchain, appending the block
corresponding to the decision to its current blockchain; $\fn{runDRC}$ will then
update the state accordingly, for instance increasing the level.
The procedure $\fn{advance}$ has one parameter,
which is optional, represented in the pseudocode as a value of an optional type
(with values of the form $\mathsf{Some}\: x$ if the parameter is present or
$\mathsf{None}$ if it is not).
The parameter is present when the current consensus instance has taken a
decision. In this case, the parameter is a tuple consisting of a block
containing the decided value and of a certificate justifying the
decision. Otherwise, when no decision is taken, the process increases its round
and filters its message buffer by removing messages no longer necessary. The
filtering procedure $\filterMsgs$ is specific to the consensus instance.

We conclude by presenting in Fig.~\ref{alg:observer} the pseudocode capturing
the behavior of the processes which are not part of a committee for a given
level. We call such processes \emph{observers}.
Contrary to committee members, observers are {passive} in the sense
that they only receive (but not send) messages and update their state
accordingly. 

\begin{figure}[t]
\begin{lstlisting}[name=Alg]
$\observer{\var{phases}[1]}$ phase:
  $\waitMsg()$
$\vdots$
$\observer{\var{phases}[m-1]}$ phase:
  $\waitMsg()$

$\observer{\var{phases}[m]}$ phase:
  $\waitMsg()$
  $\fn{advance}(\getDecision())$
\end{lstlisting}
\caption{Generic single-shot algorithm for an observer.}
\label{alg:observer}
\end{figure}

%

This observer behavior serves two purposes:
\begin{inparaenum}[i)]
\item to keep the blockchain at each observer up to date;
\item to check at the end of the round whether a decision was taken,
  and if so, whether the observer becomes a committee member at the next
  level.
\end{inparaenum}
%
To achieve i), the observer checks if it can adopt a
proposed value. It does so by invoking the $\waitMsg$ and \fn{advance}
procedures, where the parameter to \fn{advance} is obtained using the
procedure $\getDecision$, which is specific to the single-shot
consensus algorithm.
Concerning ii), when the corresponding check
(line~\ref{line:checkValidatorWaitMsg}) is successful, the observer switches
roles and acts as a committee member. We note that
line~\ref{line:checkValidatorWaitMsg} is reached when the observer end its round
and calls $\fn{advance}$, which in turn calls $\runInstance$ at the end.

\begin{remark}
  At this point, we anticipate and mention that the pseudocode a committee member executes
  has the same structure as the code of an observer.
  In particular, all processes run $\fn{advance}$ at the end of the last phase
  and thus, synchronize (line~\ref{line:sync}) in a time-based manner (depending
  on round duration), as suggested in Remark~\ref{r:sync}.
  %
\end{remark}

As for the DRC solution in Section~\ref{sec:drcbc}, the methods presented in this
section are generic. One can instantiate them by providing implementations to
the $\filterMsgs$ and $\getDecision$ procedures.
We show such concrete implementations in the next section.

\section{Single-shot \ourproto}\label{sec:algorithm}
\label{ssec:map}

To show the specific phase behavior of a committee member, we first introduce
some terminology inspired by Tezos. \ourproto committee members are called
\emph{bakers}. At each round, a value is proposed by the proposer whose turn
comes in a round-robin fashion.  \ourproto has three types of phases: \ProposeP,
\PreendorseP, and \EndorseP, each with a corresponding type of message:
$\ProposeM$ for proposals, $\PreendorseM$ for preendorsements, and $\EndorseM$
for endorsements. A fourth type of message, $\PreendorsementsM$, is for the
re-transmission of preendorsements.
A baker \emph{proposes}, \emph{preendorses}, and \emph{endorses} a
value~$v$ (at some level and with some round) when the baker
broadcasts a message of the corresponding type. Only one value per
round can be proposed or (pre)endorsed. A set of at least $2f+1$
(pre)endorsements with the same level and round and for the same value
is called a \emph{(pre)endorsement quorum certificate} (QC).
%

We consider that $\ProposeM$ messages are blocks.
This is a design choice that has the advantage that values do not have
to be sent again once decided.

Within a consensus instance, if a baker $p$ receives a preendorsement
QC for a value~$v$ and round $r$, then $p$ keeps track of $v$ as an
\emph{endorsable value} and of $r$ as an \emph{endorsable round}.
Similarly, if a baker $p$ receives a preendorsement QC for a value
$v$ and round $r$ during the \EndorseP phase of the round $r$, then $p$ locks on the
value $v$, and it keeps track of $v$ as a \emph{locked value} and of
$r$ as a \emph{locked round}.
Note that the locked round stores the most recent round at which $p$
endorsed a value, while the endorsable round stores the most recent
round that $p$ is aware of at which bakers may have endorsed a
value.

In a nutshell, the execution of a round works as follows.
During the \ProposeP phase, the designated proposer proposes a
value~$v$, which can be newly generated or an endorsable value from a
previous round $r$ of the same consensus instance.
During the \PreendorseP phase, a baker preendorses $v$ if it is not locked or if
it is locked on a value at a previous round than $r$; in particular, it does not
preendorse $v$ if it is locked and~$v$ is newly generated.
If a baker does not preendorse $v$, then it sends a
$\PreendorsementsM$ message with the preendorsement~QC that justifies
its more recent locked round.
During the \EndorseP phase, if bakers receive a preendorsement~QC
for $v$, they {lock} on it and endorse it.
If bakers receive an endorsement QC for~$v$, they \emph{decide}~$v$.

\ourproto inherits from classical BFT solutions the two voting phases
per round and the locking mechanism. Tracking endorsable values is
inherited from~\cite{tendermintv2}. \ourproto
distinguishes itself in a few aspects which we detail next.

\smallskip

\noindent\textbf{Preendorsement QCs.}
%
 For safety, bakers accept endorsable values only from higher
 rounds than their locked round. Assume a correct baker~$p$ locks and all other correct bakers
 locked at smaller rounds. Assume also that the messages from~$p$ are
 lost. To prevent~$p$ from not making progress, it is enough to
 include the preendorsement~QC that made~$p$ lock in $\EndorseM$ and
 $\ProposeM$ messages. In this way, bakers can update their endorsable values and rounds
 accordingly and propose values that can be accepted by any correct
 locked baker.
%
%
Note that Tendermint does not need such QCs because it assumes
reliable communication even in the asynchronous period.

\smallskip

\noindent\textbf{The $\PreendorsementsM$ message.}
 To achieve faster termination of a consensus instance, when a baker refuses a proposal
 because it is locked on a higher round than the endorsable round of
 the proposed value, it broadcasts a $\PreendorsementsM$ message. This
 message contains a preendorsement~QC justifying its higher locked
 round. During the next round, bakers use this QC to set their
 endorsable value to the one with the highest round. The consensus
 instance can then terminate with the first correct
 proposer. Consequently, in the worst-case scenario, i.e., when the
 first $f$ bakers are Byzantine, \ourproto terminates in $f+2$ rounds
 after $\tau$, assuming that processes have achieved round
 synchronization and that the round durations are sufficiently large.



\smallskip

\noindent\textbf{Endorsement QCs.} For processes to be able to check
that blocks received by calling $\pullChain$ are already agreed upon,
each block comes with an endorsement QC for the block at the previous
level.
Furthermore, for the same reason, in response to a pull request, a
process also attaches the endorsement QC that justifies the value in
the head of the blockchain.
Recall that a block is a $\ProposeM$ message, which contains a value
that has not yet been decided at the moment the message is sent; thus,
the endorsement QCs cannot justify the current value.

\subsection{Process state and initialization}
\label{ssec:state}
In addition to the variables mentioned in Section~\ref{sec:drcbc}, a process~$p$
running \ourproto maintains its current round~$r_p$ as well as:
\begin{itemize}
\item $\lockedBlock_p$ and $\lockedRound_p$
  to keep track respectively of the value on which~$p$ is locked and the round
  during which $p$ locked on it,
\item $\validBlock_p$ to keep track of the proposed value with
  a preendorsement QC (with the highest round), which can therefore be considered
  endorsable,
\item $\validRound_p$ and
  $\validPreendorsements_p$ to store the round and the
  preendorsement QC corresponding to an endorsable value;
\item $\headQC_p$ to store the endorsement QC for~$p$'s last decided
  value.
\end{itemize}

The variables $\validBlock_p$ and $\validRound_p$ can be recovered from the
variable $\validPreendorsements_p$. We only use them for readability.

The variable~$\headQC_p$ is the one introduced in Section~\ref{sec:drc}. We note
that, at level~1, $\headQC_p$ is empty (see Fig.~\ref{alg:drc},
line~\ref{line:updateState-first}).

%

The state of a process is initialized by the
procedure~\fn{initConsensusInstance}:
\begin{lstlisting}[name=Alg]
proc $\initInstance()$
  $r_p$ := 1
  $\lockedBlock_p$ := $\mbot$; $\lockedRound_p$ := $0$
  $\validBlock_p$ := $\mbot$; $\validRound_p$ := $0$
  $\validPreendorsements_p$ := $\emptyset$
  $\msgs_p$ := $\emptyset$
\end{lstlisting}
where, by abuse of notation, we use $x := \mbot$ to denote
that $x$ has become undefined.
%
%

%
%
\subsection{Messages and blocks}
\label{s:msg}
We write messages using the following syntax: ${\var{type}_p(\ell, r,
  h, \var{payload})}$, where $\var{type}$ is $\ProposeM$,
$\PreendorseM$, $\EndorseM$, or $\PreendorsementsM$, $p$ is the
process that sent the message, $\ell$ and $r$ are the level and the
round during which the message is generated, $h$ is the hash the block
at level~$\ell-1$, and $\var{payload}$ is the type specific content of
the message. Next, we describe the payloads for each type of message.

The payload $(\var{eQC},u,\var{eR},\var{pQC})$ of a
$\ProposeM$ message contains
the endorsement quorum~$\var{eQC}$ that justifies
the block at the previous level and the proposed value~$u$ to be
agreed on.
The payload also contains, in case~$u$ is a previously proposed value, the
corresponding endorsable round and the preendorsement QC that
justifies~$u$. If the proposed value is new, then $\var{eR}$ is $0$
and $\var{pQC}$ is the empty set.

Given a $\ProposeM(\ell,r,h, (\var{eQC},u,\var{eR},\var{pQC}))$
message, the corresponding block has contents $u$, while the remaining
fields are part of the block header.
%
%
We emphasize that the distinction between a propose message and a
block is only made for presentation purposes, to separate concerns.

\begin{remark}
  We stress that, though both the hash and the $\var{eQC}$ in a
  block~$b$ point to the predecessor of~$b$, they serve different
  purposes: the hash is needed for agreeing on the whole block while
  the $\var{eQC}$ is needed for checking that the content of the
  predecessor block was agreed upon by the right committee. Since
  certificates are not uniquely identifying blocks (but just block
  contents), certificates cannot replace hashes.
\end{remark}

The payload of a $\PreendorseM$ message consists of the value to be
agreed upon while
the payload of an $\EndorseM$ message consists of an endorsed value.
%

The payload of a $\PreendorsementsM$ message consists of a
preendorsement QC justifying some endorsable value and round.

\subsection{Message management}
\label{ssec:mm}

\ourproto messages are handled by the procedure $\handleConsensusMessage$
depicted in Fig.~\ref{alg:tb}.  Upon the retrieval of a new message $\msg$, a
process~$p$ first checks if the level, round, and hash in $\msg$'s header match
respectively $p$'s current level, either the current round or the next round,
and the hash of the block at the previous level. If yes, $p$ then checks that
the message is valid, with the procedure $\validMessage$, detailed below.
If $\msg$ is indeed valid, then it is stored in the message buffer.
Also, the variables $\validBlock_p$, $\validRound_p$, and
$\validPreendorsements_p$ are updated by the
procedure~\fn{updateEndorsable} if a preendorsement QC is observed for
a higher round than the current $\validRound_p$. This can happen if
there is already a preendorsement QC for the current proposed value in
$\msgs$ or if a preendorsement QC for a higher round than the current
$\validRound_p$ is attached to a $\ProposeM$, or $\PreendorsementsM$
message (line~\ref{line:indirectEndoUpdate}).
Finally, if the received message is from either a higher level or from the same
level but with a different hash, then $p$ attempts to resynchronize by calling
$\pullChain$. This is an optimization as the $\pullChain$ is anyway called
regularly as described in Section~\ref{sec:drcbc}.

The procedure $\validMessage$ checks the validity of each message.
$\ProposeM_q$($\ell$,$r$,$h$,($\var{eQC}$, $u$,$\var{eR}$,$\var{pQC}$))
is \emph{valid} if $q$ is the proposer for level~$\ell$ and round~$r$
and if $\var{eQC}$ is an endorsement QC for level~$\ell-1$ with the
round, hash, and value matching those in $p$'s blockchain.
%
In addition, either $\var{pQC}$ is empty and $\var{eR}$ is~$0$ (i.e.~$u$ is
newly proposed), or the round, value, and hash from $\var{pQC}$ match
$\var{eR}$, $u$, and $h_p$, respectively.
Messages in $\var{eQC}$ and $\var{pQC}$ must be valid themselves, in particular they
must be generated by bakers at levels $\ell-1$ and $\ell$,
respectively.
These validity checks ensure that the value $(u,h)$ satisfies
the $\legitimate$ predicate from Section~\ref{ssec:def}.
The validity conditions for the other types of messages are similar,
and thus omitted.
We note, however, that for preendorsements and endorsements it is
required that the corresponding proposal has been already received, so
that it can be checked that the hash included in the payload matches
the proposed value.

We highlight two additional aspects of $\handleConsensusMessage$:
\begin{itemize}
\item messages for the next round are kept (line~\ref{line:msgadd}) to
  cater for the possible clock drift;
\item messages from higher levels trigger~$p$ to ask for the sender's
  blockchain (line~\ref{line:callpc}), because such messages ``from the
  future'' suggest that $p$ is behind; however, the sender might be lying
  about being ahead.
\end{itemize}

Recall that the procedure \fn{advance} only calls $\filterMsgs()$
after a round increment. In the case of \ourproto, this procedure
removes messages not for the current round
(line~\ref{line:resetMessageSet}), thus ensuring that message buffers
are bounded.

For completeness, the helper procedures used in Fig.~\ref{alg:tb} are described
as follows:
\begin{itemize}
\item \fn{proposedValue()} returns the current proposed value.
  of the block at level~$\ell$;
\item $\fn{valueQC}(\var{qc})$ and~$\fn{roundQC}(\var{qc})$ return the
  value and respectively the round from a quorum certificate~$\var{qc}$;
\item $\fn{pQC}(\var{msg})$ returns the preendorsement QC from a $\ProposeM$ or
  $\PreendorsementsM$ message~$\var{msg}$;
\item $\proposal()$, $\pset()$, and $\vset()$ return the proposal,
  preendorsements, and respectively the endorsements contained in
  $\msgs$. We note that, thanks to filtering, the messages in the
    sets of (pre)endorsements match on the round.
\end{itemize}

\subsection{\ourproto main loop}
\label{ssec:baker}

The execution of one round of \ourproto by baker $p$, when the round's three
phases are executed in sequence, is given in Fig.~\ref{alg:tb}. We recall that
the pseudocode has the same structure as that for observers, as described in
Section~\ref{sec:gssc}.  Recall also from Section~\ref{sec:gssc} that forward
jumps can nevertheless occur when $p$ detects it is behind.

\begin{figure}[ht]
\begin{minipage}{0.53\textwidth}
\begin{lstlisting}[name=Alg,xleftmargin=-.1cm]
proc $\handleConsensusMessage(\var{msg})$
  let $\var{type}_q(\ell, r, h, \var{payload})$ = $\var{msg}$
  if $\ell = \ell_p \land h = h_p \land (r = r_p \lor r = r_p+1)$ then
    if $\validMessage(\var{msg})$ $\label{line:validMessage}$
      $\msgs_p$ := $\msgs_p \cup \set{\var{msg}}$ $\label{line:msgadd}$
      $\updateEndorsable(\var{msg})$
  if $(\ell = \ell_p \land h \neq h_p) \vee \ell>\ell_p$ then
    pullChain $\label{line:callpc}$

proc $\updateEndorsable(\var{msg})$
  if $|\pset()| \geq 2f+1$ then $\label{line:directEndoUpdate}$
    $\validBlock_p$ := $\fn{proposedValue}()$
    $\validRound_p$ := $r_p$ $\label{PrepValidRound}$
    $\validPreendorsements_p$ := $\pset()$ $\label{line:directEndoUpdateEnd}$
  else if $\fn{type}(\var{msg}) \in \{\ProposeM, \PreendorsementsM\}$ then
    $\var{pQC}$ := $\fn{pQC}(msg)$
    if $\fn{roundQC}(\var{pQC}) > \validRound_p$ then $\label{line:indirectEndoUpdate}$
      $\validBlock_p$ := $\fn{valueQC}(\var{pQC})$
      $\validRound_p$ := $\fn{roundQC}(\var{pQC})$  $\label{line:indirectEndoUpdateRound}$
      $\validPreendorsements_p$ := $\var{pQC}$

proc $\filterMsgs()$
  $\msgs_p$ := $\msgs_p \setminus $
           $\set{\var{type}(\ell,r,h,\var{payload})\in\msgs_p\mid r \neq r_p}$  $\label{line:resetMessageSet}$
  \end{lstlisting}
\end{minipage}
\begin{minipage}{.55\textwidth}
\begin{lstlisting}[name=Alg]
 $\ProposePL$ phase:
  if $\proposer(\ell_p, r_p) = p$ then $\label{ppropCheck3}$
    $u$ := if $\validBlock_p \neq \mbot$ then $\validBlock_p$
         else $\getBlock()$ $\label{ppropValueToSend}$
    $\var{payload}$ := $(\headQC_p, u,$
                $\validRound_p, \validPreendorsements_p)$
    broadcast $\ProposeM_p(\ell_p, r_p, h_p, \var{payload})$ $\label{ppropCheck3End}$
  $\waitMsg()$

 $\PreendorsePL$ phase:
  if $\exists q, \var{eQC}, u, \var{eR}, \var{pQC}:$
     $\ProposeM_q(\ell_p, r_p, h_p, (\var{eQC}, u, \var{eR}, \var{pQC})) \in \msgs_p\,\land$ $\label{propRcv}$
     $(\lockedBlock_p = u \vee \lockedRound_p < \var{eR} < r_p)$ then $\label{propCheckLock}$ $\label{propCheck1}$
    broadcast $\PreendorseM_p(\ell_p, r_p, h_p, \hash(u))$ $\label{propCheck1End}$
  else if $\lockedBlock_p \neq \bot$ then $\label{preendBc}$
    broadcast $\PreendorsementsM(\ell_p, r_p, h_p,$
                                    $\validPreendorsements_p)$ $\label{preendBc2}$
  $\waitMsg()$

 $\EndorsePL$ phase:
  if $|\pset()| \geq 2f + 1$ then $\label{Endorse}$
    $u$ := $\fn{proposedValue}()$
    $\lockedBlock_p$ := $u$; $\label{EndorseLock}$ $\lockedRound_p$ := $r_p$ $\label{EndorseLockRound}$
    broadcast $\EndorseM_p(\ell_p, r_p, h_p, \hash(u))$ $\label{EndorseEnd}$
    broadcast $\validPreendorsements_p$
  $\waitMsg()$
  $\fn{advance}(\getDecision())$ $\label{line:callAdvance}$
\end{lstlisting}
\end{minipage}
\caption{Single-shot \ourproto for baker~$p$ (right) and message management (left).}
  \label{alg:tb}
\end{figure}


Each phase consists of a conditional broadcast followed by a call to
$\waitMsg$ (described in~Section~\ref{sec:drcbc}). In addition,
the~\EndorseP phase calls \fn{advance} (described
in~Section~\ref{sec:drcbc}). Next, we describe the conditions to
broadcast in each phase.

In the \ProposeP phase, $p$ checks if it is the proposer for the
current level~$\ell_p$ and round~$r_p$ (line~\ref{ppropCheck3}). If
so, $p$ proposes:
\begin{itemize}
  \item either a new value~$u$, returned by the procedure $\getBlock$;
    here it is assumed that $u$ is consistent with respect to the
    value~$u'$ contained in the last block of the blockchain of the
    process that calls this procedure; that is, $\consistent(v,v')$
    holds (see Section~\ref{ssec:def}), where $v,v'$ are the output
    values corresponding to $u,u'$;
  \item or its~$\validBlock_p$ if defined; in this case,~$p$ includes in the
    payload of its proposal the corresponding endorsable round and the
    preendorsement QC that justifies it.
\end{itemize}
The payload also includes the endorsement QC to justify the decision
for the previous level.

In the \PreendorseP phase, $p$ checks if the value~$u$ from the
$\ProposeM$ message received from the current proposer is
preendorsable (lines \ref{propRcv}-\ref{propCheck1}). Namely, it
checks whether one of the following conditions are satisfied:
\begin{itemize}
\item $p$ is unlocked ($\lockedRound_p = 0$, thus the second disjunction at
  line~\ref{propCheck1} is true); or
\item $p$ is locked (i.e.~$\lockedRound_p > 0$), $u$ was already
  proposed during some previous round (i.e.~$0<\var{eR}<r_p$), and:
  \begin{itemize}
  \item $p$ is already locked on $u$ itself (thus the first disjunction at
    line~\ref{propCheck1} is true); or
  \item $p$ is locked on $u'\neq u$ and its locked round is smaller than the
    endorsable round associated to~$u$.
  \end{itemize}
\end{itemize}
In the second case, we know there is a preendorsement QC
for~$u$ and round $\var{eR}$, thanks to the validity check on the
$\ProposeM$ message.
If the condition holds, then~$p$ preendorses~$u$.
If~$p$ cannot preendorse~$u$ because it is locked on some
value~$u'\neq u$ with a higher locked round than $\var{eR}$, then $p$
broadcasts the preendorsement QC that justifies~$v'$. If received on
time, this information allows the next correct proposer to propose a
value that passes the above checks for all correct bakers.

In the \EndorseP phase, $p$ checks if it received a preendorsement QC
for the proposed value~$u$. If yes, $p$ updates its $\lockedBlock$ and
$\validBlock$ and broadcasts its $\EndorseM$ message, along with all
the $\PreendorseM$ messages for $u$ (lines
\ref{Endorse}-\ref{EndorseEnd}).
Note also that in this case $p$ has already updated its endorsable
value to~$u$ and its endorsable round to~$r_p$ while executing
$\waitMsg$.

Finally, at the end of this last phase, which is also the end of the
round, bakers call $\fn{advance}$ with a parameter that signals
whether a decision can be taken or not. This parameter is obtained
using $\fn{getDecision}$, implemented is as follows:
\begin{lstlisting}[name=Alg]
proc $\getDecision()$
  if $|\vset()| \geq 2f+1$ then
    return $\mathsf{Some}$ $(\proposal(),\vset())$ $\label{line:quorumCond}$
  else
    return $\mathsf{None}$
\end{lstlisting}

\subsection{The $\betterChain$ procedure}

The remaining pieces needed for a complete instantiation of the
procedures in Section~\ref{sec:drc} are the implementations
of the~$\betterChain$ and~$\fn{getCertificate}$ procedures.

We recall that~$\betterChain$ takes as parameters a $\var{chain}$ and
$\pc$, both returned by $\pullChain$. We recall also that $\pc$ is
either a proposal at the current level or just a certificate if the
process answering to the $\pullChain$ request has not seen a valid
proposal yet. The role of~$\betterChain$ is to make processes agree on
the same blockchain head; recall that they already agree on the head
contents, but not necessarily on the head's header. Agreeing on the
same blockchain head has in turn two roles:
\begin{itemize}
\item allowing agreement on the round at the which a decision was
  taken at the previous level, which is one of the ingredients for
  processes to synchronize at the current level, as explained in
  Section~\ref{sec:sync}.
\item allowing agreement to take place at the current level; recall
  that at the current level agreement needs to be reached also on the
  hash of the block at the predecessor level, that is, on the hash of
  the head of a process' blockchain.
\end{itemize}

To reach these two goals, as suggested in Section~\ref{sec:sync},
processes adopt the head with the smallest round. However, there is a
caveat: if this would be the only check done by $\betterChain$,
processes might end up with a head on top of which no proposal will be
accepted in case they have seen an endorsable value: indeed, the hash component
of such a value may not match the new head.
To avoid this situation, a process first performs an additional check
in case they have seen an endorsable value. When $\pc$ is a proposal, the check is
similar to the check for preendorsing (line~\ref{propCheckLock}): the
endorsable round of process~$p$ is smaller than the one in the
received proposal (line~\ref{line:betterScore}). When $\pc$ is a
certificate we simply required that the process has not seen an endorsable value.
The $\betterChain$ procedure implementing these checks is given next.

\begin{lstlisting}[name=Alg]
proc $\betterChain(\var{chain}, \pc)$
  let $\block{\ell,r,\dots}{\cdot}$ = $\fn{head}(\var{chain})$
  match $\pc$ with
  | $\ProposeM(\_, \_, \_, (\_, \_, \var{eR}, \_))$ ->
    return $\validRound_p < \var{eR} \vee \big(\validRound_p = \var{eR} \wedge r < \decisionRound(\ell_p-1)\big)$ $\label{line:betterScore}$
  | _ -> # $\textcolor{blue}{\pc}$ is a certificate
    return $\validRound_p = 0 \wedge r < \decisionRound(\ell_p-1)$ $\label{line:betterScore2}$
\end{lstlisting}
In the pseudocode, $\block{\dots}{\dots}$ denotes a block, with the
part before the semicolon representing the block's header and the part
after it its contents.
The procedure $\fn{head}(\var{chain})$ returns the head of~$\var{chain}$. Also, recall that
$\decisionRound(\ell)$ returns the round contained in the header of
the block at level~$\ell$ in the caller's blockchain.

As for the implementation of~$\fn{getCertificate}$, it is a simple
match on~$\pc$:
\begin{lstlisting}[name=Alg]
proc $\fn{getCertificate}(\pc)$
  match $\pc$ with
    | $\ProposeM_q(\_, \_, \_, (\var{eQC}, \_, \_, \_))$ -> return $\var{eQC}$
    | $\var{eQC}$ -> return $\var{eQC}$
\end{lstlisting}

%


\section{Correctness and complexity}
\label{sec:cc}

The following theorem states that \ourproto provides a solution to DRC.
Its proof can be found in Appendix~\ref{sec:correctness}.
\begin{theorem}
  \ourproto satisfies validity, agreement, and progress.
\end{theorem}

\subsection{Bounded memory.}
We assume that all values referred to by global or local variables of
a process~$p$ are stored in volatile memory, except for the
variable~$\blockchain_p$ whose value is stored on disk.

The following lemma shows that a process can use fixed-sized buffers,
namely of size~$4n$. We recall that the message buffer is
represented by the $\msgs_p$ variable.
\begin{lemma}\label{l:bounded-buffers}
  For any correct process~$p$, at any time, $|\msgs_p|\leq 4n+2$.
\end{lemma}
\begin{proof}
  Let $p$ be some correct process.
  Given that in $\msgs_p$ only messages from the current and next round are
  added (line~\ref{line:msgadd}), and that with each new round messages from the
  previous round are filtered out (line~$\ref{line:resetMessageSet}$), $\msgs_p$
  contains at most $2$ proposals, at most $2n$ preendorsements, and at most $2n$
  endorsements.
  \qed
\end{proof}

The following result states that a process only uses bounded memory.
\begin{theorem}\label{thm:bounded-memory}
  At any time, the size of the volatile memory of any correct process
  is in~$\mathcal{O}(n)$.
\end{theorem}
\begin{proof}
  A correct process maintains a constant number of variables, and
  except $\msgs$, each variable stores a primitive value or a QC. A QC
  contains at most $n$ messages and each message has a constant
  size. The $\mathcal{O}(n)$ bound follows from these observations, and
  the observation concerning the $\msgs$ variable from the proof of
  Lemma~\ref{l:bounded-buffers}.
  \qed
\end{proof}

\subsection{Message and Round Complexity.}
Each round has a message complexity of $O(nm)$ due to the $n$-to-$m$
broadcast, where $m$ is the current number of processes in the system.

Concerning round complexity, it is known that consensus, in the worst case
scenario, cannot be reached in less than $f+1$ rounds~\cite{fischer1982lower}.
In \ourproto, after bakers synchronize and the round durations are sufficiently
long (namely, at least $\delta+2\rho)$, a decision is taken in at most $f+2$
rounds, as already mentioned in Section~\ref{ssec:map}. See
Lemma~\ref{l:allSynchronized} in Section~\ref{sec:correctness} for a proof.
Intuitively, $f$ rounds are needed in case the proposers of these rounds are
Byzantine. Another round is needed if there is a correct process locked on a
higher round than the endorsable round of the proposed value. However, in this
case, the next proposer is correct and will have updated its endorsable round,
and therefore its proposed value will be accepted and decided by all correct
processes.

The number of rounds necessary for the round duration to become larger
than~$\delta$ depends on $\roundDur$'s growth.
For instance, if $\roundDur$ grows exponentially, then this number is in
$O(\log(\delta+2\rho))$.
(This point is related to the recommendation for the choice of~$\Delta$ in
Remark~\ref{remark:choosingDelta}, see item~(ii).)

\begin{remark}
  The space required by a QC may be reduced by using threshold signatures which
  has the effect of reducing the message size from $O(n)$ to $O(1)$.  Note
  however that this technique requires threshold keys to be generated a priori,
  for example using a distributed key generation algorithm.

  Since knowledge of the committee participants and their public keys is known a
  priori, it is possible to use aggregated signatures formed by signing with
  standard keys, along with a bitfield which represents the presence or absence
  of a participant's signature.
  Then aggregated signatures can be used instead of threshold
  signatures with similar effect besides the extra space required for
  the bitfield.

  The use of threshold signatures can be combined with a restructuring
  of the communication pattern within a round to also reduce the
  message complexity, as done for instance in
  HotStuff~\cite{hotstuff}: processes send their (pre)endorsements to
  the proposer, who combines the received signature shares into one
  threshold~signature.
\end{remark}



\subsection{Recovery time}
\label{sec:rc}

We analyze the time that processes need in order to recover from the
asynchronous period, that is, the time to synchronize with each other
after~$\tau$ and start a new round synchronized.

For simplicity, in this analysis, we assume that there is no clock drift
after~$\tau$, that is~$\rho=0$.
We say that two correct processes are \emph{synchronized} if they are at the
same level, round, and phase.
%
%
Let $\tau_{\rc}>\tau$ be the time of the beginning of the first round at which
all correct processes are synchronized (already at the beginning of that round).
%
We define the \textit{recovery time} $\Delta_{\rc}$ as $\tau_{\rc}-\tau$.

To give an upper bound on $\Delta_{\rc}$, we first introduce some notation.
Let $\Delta_{err}$ be the bound on the clock error that can occur before~$\tau$,
i.e.,~the maximum value of $|t-\fn{now}()|$, over all real clock values $t<\tau$,
where $\fn{now}$ is called at~$t$. Let $\Delta_{pull}(\ell)$ be the maximum delay
between an invocation of $\pullChain$ after~$\tau$ and the reception of the new
chain, where~$\ell$ is the number of received blocks. In other words,
$\Delta_{pull}(\ell)$ is the maximum time that a process needs to fetch $\ell$
missing blocks. 
Note that this is at least one round-trip time: the time to ask for the current
blockchain and to get the reply.
We also define $\ell_\tau$ to be the maximal level at which a correct process can be
at~$\tau$. Let $t^1_{\ell_\tau+1}$ be the end time of the first round at which
the decision for the block at level $\ell_\tau$ has been taken. Note that this
is also the start time of the first round at level $\ell_\tau+1$.
Lastly, we let $\roundFromTime$ be the function that, given a time
difference~$\var{td}$, returns the round at which a process would be at time $t$
had it started its consensus instance $\var{td}$ time ago. Formally,
$\roundFromTime(\var{td}) = r$ iff $\var{td} \in [s_r,s_{r+1})$, where
  $s_1=0$ and $s_r=\sum_{i=1}^{r-1}\roundDur(i)$ for $r>1$.


The upper bound on the recovery time depends on (a) $\Delta_{err}$, (b) the time
interval~$I$ at which processes call the $\pullChain$ primitive; (c) the time a
process needs to fetch the missing blocks, which, in the worst case, when at
$\tau$ the slowest process is still at the genesis block, is
$\Delta_{pull}(\ell_\tau)$, and (d) the round duration function~$\Delta$.

\begin{theorem}
  The recovery time $\Delta_{\rc}$ is upper bounded by
  $$\max\left(
  \Delta_{err},\
  I + \Delta_{pull}(\ell_\tau) + \roundDur(r),\
  \roundDur(r') + \roundDur(r'+1)
  \right),$$
  where $r = \roundFromTime(\tau+I+\Delta_{pull}(\ell_\tau)-t^1_{\ell_\tau+1})$ and
  $r'=\roundFromTime(\tau-t^1_{\ell_\tau+1})$.
\end{theorem}

The terms $\Delta(r)$ and $\Delta(r'+1)$ express that a process may need, even
after being synchronized, to further wait until a new round begins.
We note that the presence of these terms justifies the recommendation for the
choice of~$\Delta$ in Remark~\ref{remark:choosingDelta}, point~(iii): minimizing
the round durations also minimizes $\Delta_{\rc}$.

While we can control the interval~$I$ and the round duration function~$\Delta$,
the other two variables (namely $\Delta_{err}$ and $\ell_{\tau}$) are exogenous
and thus cannot be controlled.
However, we believe that in practice the time to pull a new chain (and even to
pull just the last block) is considerably bigger than the maximum error clock
that a process can experience during the asynchronous period.

%

\begin{proof}
  We call $\cp_p = \langle \ell_p, r_p \rangle$ the computation position of a
  process~$p$.
  We define the \textit{correct computation position}~$\cp^*$ as the level and
  round that any correct process would be at, if it had the current longest
  valid chain~$c$ and access to a synchronized clock. That is, $\cp^* = \langle
  \level^*, \round^* \rangle$, where $\level^*$ is the length of~$c$ and
  $\round^*$ is computed by $\fn{synchronize}$ when called over~$c$ and when
  $\fn{now}()$ returns the real time.

  We distinguish the following cases in which a correct process~$p$ can be at
  time~$\tau$:
  \begin{enumerate}[(A)]
  \item $p$ is at a lower level than $\cp^*$, i.e.~$\ell_p<\level^*$;
  \item $p$ is at the same level as $\cp^*$, but at a smaller or equal round,
    i.e.~$\ell_p=\level^*$ and $r_p \leq \round^*$; and,
  \item $p$ is at the same level as $\cp^*$, but at a greater round,
    i.e.~$\ell_p=\level^*$ and $r_p>\round^*$.
  \end{enumerate}
  Note that, as we consider $\cp^*$ at $\tau$, we have $\level^*=\ell_\tau+1$.

  \smallskip

  \noindent {\em Case (A).}
  %
  We assume the worst case, namely, that at time $\tau$ the local blockchain of $p$
  contains only the genesis block.
  Recall that process $p$ is calling $\pullChain$ every $I$ time units. We
  assume that before~$\tau$ all messages to and from $p$ are lost and that $p$
  calls $\pullChain$ at time $\tau-\epsilon$, for some $\epsilon>0$, without
  being able to update its local blockchain. When $p$ calls $\pullChain$ at time
  $\tau - \epsilon + I$ this operation succeeds. Thus, at time $\tau - \epsilon
  + I + \Delta_{pull}(\ell_\tau)$ the local blockchain of $p$ is updated up to
  level $\ell_\tau$, and at the subsequent invocation of $\fn{synchronize}()$,
  also at time $\tau - \epsilon + I + \Delta_{pull}(\ell_\tau)$, process~$p$
  finally gets to the correct computation position~$\cp^*$. Let $r_p$ and
  $\dt_p$ be the values returned by $\fn{synchronize}()$. If $ \dt_p\neq 0$, $p$
  is not at the very beginning of the round, thus it must wait one more round to
  reach $\tau_{\rc}$.
  That is, $p$ has to wait $\roundDur(r_p)$ further.
  Thus, in this scenario, $p$ reaches the beginning of the round with the
  correct computation position at most at time $\tau + I +
  \Delta_{pull}(\ell_\tau)+\roundDur(r_p)$. Note that $r_p =
  \roundFromTime(\tau+I+\Delta_{pull}(\ell_\tau)-t^1_{\ell_\tau+1})$.
  %
  %

  \smallskip

  \noindent{\em Case (B).} Since, at time $\tau$, $\cp_{p}$ is behind
  $\cp^*$, then as soon as $p$ calls $\fn{synchronize}()$ after~$\tau$, in the
  worst case at time $\tau - \epsilon' + \roundDur(r_p)$, for some
  $\epsilon'>0$, then it reaches $\cp^*$. Here $r_p =
  \roundFromTime(\tau-t^1_{\ell_\tau+1})$. However, $p$ may reach $\cp^*$ in the
  middle of a round, therefore $p$ reaches the beginning of the first round
  after the synchronization (with the correct computation position) at the
  latest at time $\tau + \roundDur(r_p) + \roundDur(r_p+1)$.

  \smallskip

  \noindent{\em Case (C).} Finally, if $p$'s computation position $\cp_{p}$ is
  in advance with respect to $\cp^*$ then~$p$ has to wait for at most the maximum
  clock error $\Delta_{err}$ to get to $\cp^*$. Note that in this case, it
  reaches $\cp^*$ at the very beginning of the round $\round^*$.

  \smallskip

  To conclude, it suffices to note that $\Delta_{\rc}$ is upper bounded by the
  maximum between the terms in the above cases.
  \qed
\end{proof}


\section{Conclusion}\label{sec:conclusion}

In this paper, we proposed a formalization of dynamic repeated consensus,
a general approach to solve it, and a BFT solution working with bounded
buffers by leveraging a blockchain-based synchronizer.


As future work, we see several exciting directions:
experimentally evaluate the provided solution;
explore the relationship between achieving asynchronous responsiveness and
providing bounded buffers;
improve message size and complexity by means of aggregated or threshold
signatures;
mechanize the proofs;
and analyze \ourproto from an economic perspective when considering rational
agents.

\subsubsection*{Acknowledgements}
We thank Edward Tate for his help in writing the remark on multisignatures.
We thank Philippe Bidinger for feedback on a previous version of this report.

\bibliographystyle{abbrv}
\bibliography{biblio}

\begin{thebibliography}{10}

\bibitem{tezos-tutorial}
V.~Allombert, M.~Bourgoin, and J.~Tesson.
\newblock {Introduction to the Tezos Blockchain}.
\newblock In {\em Proc. High Performance Computing and Simulation}, 2019.

\bibitem{TendermintCorrectness}
Y.~Amoussou-Guenou, A.~Del~Pozzo, M.~Potop-Butucaru, and S.~Tucci-Piergiovanni.
\newblock Correctness of {Tendermint}-core blockchains.
\newblock In {\em Proc. Principles of Distributed Systems}, 2018.

\bibitem{dissecting}
Y.~Amoussou{-}Guenou, A.~Del~Pozzo, M.~Potop{-}Butucaru, and S.~{Tucci
  Piergiovanni}.
\newblock Dissecting {Tendermint}.
\newblock In {\em Proc. Networked Systems}, 2019.

\bibitem{smrdyn}
A.~Bessani, E.~Alchieri, J.~Sousa, A.~Oliveira, and F.~Pedone.
\newblock From byzantine replication to blockchain: Consensus is only the
  beginning, 2020.

\bibitem{tendermintv2}
E.~Buchman, J.~Kwon, and Z.~Milosevic.
\newblock The latest gossip on {BFT} consensus.
\newblock {\em CoRR}, 2018.

\bibitem{pbft}
M.~Castro and B.~Liskov.
\newblock Practical {Byzantine} fault tolerance and proactive recovery.
\newblock {\em ACM Trans. Comput. Syst.}, 2002.

\bibitem{pala}
T.~H. Chan, R.~Pass, and E.~Shi.
\newblock Pala: {A} simple partially synchronous blockchain.
\newblock {\em {IACR} Cryptol. ePrint Arch.}, 2018.

\bibitem{algorand}
J.~Chen and S.~Micali.
\newblock Algorand: {A} secure and efficient distributed ledger.
\newblock {\em Theor. Comput. Sci.}, 2019.

\bibitem{redbelly17}
T.~Crain, V.~Gramoli, M.~Larrea, and M.~Raynal.
\newblock ({Leader}/\allowbreak {Randomization}/\allowbreak {Signature})\--free
  {Byzantine} consensus for consortium blockchains, 2017.

\bibitem{delporte2008finite}
C.~Delporte-Gallet, S.~Devismes, H.~Fauconnier, F.~Petit, and S.~Toueg.
\newblock With finite memory consensus is easier than reliable broadcast.
\newblock In {\em Proc. Principles of Distributed Systems}, 2008.

\bibitem{DLS88}
C.~Dwork, N.~A. Lynch, and L.~J. Stockmeyer.
\newblock Consensus in the presence of partial synchrony.
\newblock {\em J. {ACM}}, 1988.

\bibitem{fischer1982lower}
M.~J. Fischer and N.~A. Lynch.
\newblock A lower bound for the time to assure interactive consistency.
\newblock {\em Information processing letters}, 1982.

\bibitem{GarayKL15}
J.~A. Garay, A.~Kiayias, and N.~Leonardos.
\newblock The bitcoin backbone protocol: Analysis and applications.
\newblock In {\em Proc. EUROCRYPT International Conference}, 2015.

\bibitem{tezos-white-paper}
L.~Goodman.
\newblock Tezos -- a self-amending crypto-ledger, 2014.

\bibitem{dfinity}
T.~Hanke, M.~Movahedi, and D.~Williams.
\newblock {DFINITY} technology overview series, consensus system.
\newblock {\em CoRR}, 2018.

\bibitem{cosmos}
J.~Kwon and E.~Buchman.
\newblock {Cosmos: A Network of Distributed Ledgers}.

\bibitem{Merrit13}
M.~Merritt and G.~Taubenfeld.
\newblock Computing with infinitely many processes.
\newblock {\em Inf. Comput.}, page 12–31, Dec. 2013.

\bibitem{bitcoin}
S.~Nakamoto.
\newblock {Bitcoin: A Peer-to-Peer Electronic Cash System}, 2008.

\bibitem{lumiere}
O.~Naor, M.~Baudet, D.~Malkhi, and A.~Spiegelman.
\newblock Cogsworth: Byzantine view synchronization.
\newblock {\em CoRR}, 2019.

\bibitem{Keidar20}
O.~Naor and I.~Keidar.
\newblock Expected linear round synchronization: The missing link for linear
  byzantine {SMR}.
\newblock {\em CoRR}, 2020.

\bibitem{EmmyPlus_blog_post}
{Nomadic Labs}.
\newblock {Analysis of Emmy\textsuperscript{+}}.
\newblock \url{https://blog.nomadic-labs.com/analysis-of-emmy.html}, 2019.

\bibitem{PS18}
R.~Pass and E.~Shi.
\newblock Rethinking large-scale consensus.
\newblock {\em {IACR} Cryptol. ePrint Arch.}, 2018.

\bibitem{ibft2}
R.~Saltini.
\newblock Correctness analysis of {IBFT}.
\newblock {\em ArXiv}, 2019.

\bibitem{Schneider90}
F.~B. Schneider.
\newblock Implementing fault-tolerant services using the state machine
  approach: A tutorial.
\newblock {\em ACM Comput. Surv.}, page 299–319, 1990.

\bibitem{shahmirzadi2009relaxed}
O.~Shahmirzadi, S.~Mena, and A.~Schiper.
\newblock Relaxed atomic broadcast: State-machine replication using bounded
  memory.
\newblock In {\em Proc. IEEE International Symposium on Reliable Distributed
  Systems}, 2009.

\bibitem{libraSMR}
{The LibraBFT Team}.
\newblock State machine replication in the {Libra} blockchain, 2019.

\bibitem{hotstuff}
M.~Yin, D.~Malkhi, M.~K. Reiter, G.~G. Gueta, and I.~Abraham.
\newblock {HotStuff}: {BFT} consensus with linearity and responsiveness.
\newblock In {\em Proc. ACM Symposium on Principles of Distributed Computing},
  2019.

\end{thebibliography}

\appendix

\section{Correctness proof}\label{sec:correctness}

In this section we prove the correctness of \ourproto. The proofs for
the agreement and validity properties are inspired
by~\cite{dissecting}.



\subsection{Validity and Agreement}

\begin{theorem}\label{l:validity}
  \ourproto satisfies validity.
\end{theorem}

\begin{proof}
The local chain of a correct process~$p$ is formed by proposals and/or
chains obtained by~$p$ calling $\pullChain$. In either case, the
content of each block satisfies the predicate~$\isValid$ by the
definition of either $\validMessage$ or $\validChain$.
  %
  \qed
\end{proof}

\begin{lemma}\label{l:bcOnce}
  Correct bakers only preendorse and endorse at most once per round at
  a given level.
\end{lemma}
\begin{proof}
  $\PreendorseM$ and $\EndorseM$ messages are sent only during the corresponding
  phase (line~\ref{propCheck1End} and line~\ref{EndorseEnd}, respectively). To
  show that there is at most one $\PreendorseM$, respectively one $\EndorseM$
  per round it suffices to show that a phase is executed only once for a given
  round. We only need to make two observations. Firstly, phases are executed
  sequentially: \ProposeP, then \PreendorseP, then \EndorseP. Secondly,
  non-sequential jumps happen only at line~\ref{line:jumpToPhase} (respectively
  at line~\ref{line:jumpToObserverPhase}) in $\runInstance$; in turn,
  $\runInstance$ is called by either:
  \begin{itemize}
    \item \fn{advance} (line~\ref{line:runInstance}), after increasing the
      round; or
    \item \fn{runDRC} (line~\ref{l:ri}), after increasing the level
      (line~\ref{line:nextLevel}) once a decision is taken (line~\ref{l:reta})
      or a longer chain is received (line~\ref{l:ret}).
  \end{itemize}
  Therefore a phase is never executed twice for the same round.  \qed
\end{proof}

\begin{lemma}\label{l:oneMaj}
  At most one value can have a (pre)endorsement QC per round.
\end{lemma}
\begin{proof}
	We reason by contradiction. Let $v$ and $v'$ be two different values. If
    both $v$ and $v'$ have a (pre)endorsement QC then by quorum intersection at
    least $1$ correct process (pre)endorses both~$v$ and~$v'$, which contradicts
    Lemma~\ref{l:bcOnce}.
    %
    \qed
\end{proof}

We say a baker~$p$ is locked on a tuple $(u,h)$ if
$\lockedBlock_p=u$ and $h_p=h$.
We define $L^{u,h}_{\ell,r}$ as the set of \emph{correct} bakers
locked on the tuple~$(u,h)$ at level~$\ell$ and at the end of
round~$r$.
%
%
We also define $\pendos(\ell,r,u,h)$ as the set of preendorsements
generated by \emph{correct} processes for some level~$\ell$, some
round~$r$, some value~$u$, and some hash~$h$. Note that this set may
be empty.
Note also that if there is a preendorsement QC for $(\ell,r,u,h)$,
then $|\pendos(\ell,r,u,h)|\geq f+1$. This is because a preendorsement
QC contains at most $f$ preendorsements generated by Byzantine
processes. And vice-versa, if $|\pendos(\ell,r,u,h)|\leq f$, then
there cannot be a preendorsement QC for $(\ell,r,u,h)$.

\begin{lemma}\label{l:preendorseBound}
  Let $\ell$ be a level, $r$ a round, $u$ a value, and $h$ a block hash.
  For any round $r'\geq r$ and any tuple $(u',h')\neq (u,h)$, if
  $|L^{u,h}_{\ell,r}| \ge f+1$, then $|\pendos_p(\ell,r',u',h')|\leq
  f$.
\end{lemma}

\begin{proof}
  We reason by contradiction.
  Suppose that $|L^{u,h}_{\ell,r}| \ge f+1$, and let $r'\geq r$ be the
  smallest round for which there exists
  a tuple~$(u',h')\neq (u,h)$ such that
  $|\pendos(\ell,r',u',h')|\geq f+1$.
  %
  %
  As $|L^{u,h}_{\ell,r}| \ge f+1$ and $|\pendos(\ell,r',u',h')|\geq
  f+1$, there is at least one correct process $p$ such that $p\in
  L^{u,h}_{\ell,r}$ and $p$ preendorsed $(u',h')$ at round~$r'$.
  As $p\in L^{u,h}_{\ell,r}$, we have that $p$ is locked on~$(u,h)$ at round~$r$.
  Since $p$ preendorsed (line \ref{propCheck1End}) at round~$r'$, it means that
  one of the two disjunctions at line~\ref{propCheckLock} holds.
  Note that the value of $r_p$ at line~\ref{propCheckLock} is $r'$
  in this case.

  Suppose the first disjunction holds, i.e.,
  $\lockedBlock_p=u'$. As a process can re-lock only in the
  phase~\EndorseP, under the condition at line~\ref{Endorse}, this
  means that there is a round $r''$ with $r\leq r''<r'$ and at which
  $|\pset()|\geq 2f+1$. Therefore $|\pendos(\ell,r'',u',h')|\geq
  f+1$. This contradicts the minimality of~$r'$.

  Suppose now that the second disjunction holds, that is,
  $\lockedRound_p < r'' < r'$
  where the round~$r''$ is the endorsable round of the proposer of~$u'$.
  We note that a process cannot unlock (i.e.~unset
  $\lockedRound$), but only re-lock (i.e.~set $\lockedRound$ to a
  different value). Therefore $\lockedRound_p\geq r$ at round $r'$ and
  from this, we obtain that $r'' > r > 0$. From the validity
  requirements of a propose message, we obtain that it contains a
  preendorsement QC for~$(u',h')$. Thus we have that
  $|\pendos(\ell,r'',u',h')|\geq f+1$.
  This contradicts the minimality of $r'$, since $r''<r'$.
  %
  \qed
\end{proof}




\begin{lemma}\label{l:agreement}
  No two correct processes have two different committed blocks at the
  same level in their blockchain.
\end{lemma}

\begin{proof}
  We reason by contradiction.
  Let $\ell$ be some level.
  Assume that two different correct processes~$p$, $p'$ have
  respectively two different committed blocks~$b,b'$ at level~$\ell$ in their
  blockchain, with $b\neq b'$.
  %

  By the definition of committed blocks (Section~\ref{sec:drc}), as~$b$
  is a committed block at~$\ell$, the level of the head of
  $p$'s blockchain is at least $\ell+1$.
  Then, as $p$ has a block at level~$\ell+1$ in
  his blockchain, $p$ has observed an endorsement QC for
  $(\ell+1,r,h,u)$ for some value~$u$ and some round~$r$, where $h$ is
  the hash of block~$b$.  Similarly, $p'$ has observed an endorsement~QC
  for $(\ell+1,r',h',u')$ for some value~$u'$ and some round~$r'$,
  where $h'$ is the hash of block~$b'$. As $b\neq b'$, we have that
  $h\neq h'$, therefore $(u,h)\neq(u',h')$.
  We assume without loss of generality that $r\leq r'$.

  Since there are at most $f$ Byzantine processes, and by
  Lemma~\ref{l:bcOnce} correct bakers can only endorse once per
  round, it follows that at least $f+1$ correct bakers
  endorsed~$(u,h)$ during round~$r$ at level~$\ell$.
  Before broadcasting an endorsement for~$(u,h)$ at round~$r$
  (line~\ref{EndorseEnd}) any correct process sets its $\lockedBlock$
  to $u$ and its $\lockedRound$ to $r$ (line~\ref{EndorseLock}),
  thus~$|L^{u,h}_{\ell,r}| \geq f+1$.
  %
  %
  By Lemma~\ref{l:preendorseBound}, since $|L^{u,h}_{\ell,r}| \ge f+1$, we
  also have $|\pendos(\ell,r'',u'',h'')|\leq f$, for
  any round $r''\geq r$, and any value $u''$ with $(u'',h'')\neq (u,h)$.
  This means that a correct process cannot endorse some $(u'',h'')\neq
  (u,h)$ at a round~$r''\geq r$. This in turn means that there cannot
  be $2f+1$ endorsements for~$(u'',h'')\neq (u,h)$ with round~$r''\geq r$.
  %
  %
  This contradicts the fact that there is a QC for $(\ell+1,r',u',h')$.
  \qed
\end{proof}

\begin{theorem}\label{thm:agreement}
  \ourproto satisfies agreement.
\end{theorem}
\begin{proof}
  We reason by contradiction. Suppose there are two correct processes
  $p$ and $p'$ with outputs $\bar{v}$ and $\bar{v}'$ such that neither
  one is the prefix of the other.
  This means that there is a level~$\ell$ such that
  $\bar{v}[\ell]\neq \bar{v}'[\ell]$.  We have that processes $p$ and
  $p'$ have two different committed blocks at level~$\ell$, which
  contradicts Lemma~\ref{l:agreement}.
  \qed
\end{proof}

\subsection{Progress}

Let $\Phases$ be the set of labels \ProposeP, \PreendorseP, and
\EndorseP.
Let $S_p:\pNat\times\pNat\times\Phases\to\mathbb{R}$ be the function
such that $S_p(\ell,r,\var{phase})$ gives the starting time of the
phase~$\var{phase}$ of round~$r$ of process~$p$ at level~$\ell$.
We consider that the function $S_p$ returns the real time, not the local time of
process~$p$. Note that for different processes $p$ and $q$, the function $S_p$
and $S_q$ may return different times for the same input, because $p$ and $q$
determine the starting time of their phases based on their local clocks, which
may be different before~$\tau$.

Contrary to Section~\ref{sec:rc}, we consider the general case when $\rho \geq
0$. We say that two correct processes $p$ and $p'$ are \emph{synchronized} if
$\ell_p=\ell_{p'}$, $|r_p - r_{p'}| \leq 1$, and $|S_p(\ell_q,r_q,\var{phase}_q)
- S_{p'}(\ell_q,r_q,\var{phase}_q)| \leq 2\rho$, where $q\in\set{p,p'}$ is the
process which is ``ahead''.
We say that $q$ is \emph{ahead} of $q'$ (or that $q'$ is \emph{behind}
$q$) if $S_{q}(\ell_q,r_q,\var{phase}_{q}) \leq
S_{q'}(\ell_q,r_q,\var{phase}_q)$.
%
%
Intuitively, two processes are synchronized if they are roughly at the
same level, round, and phase, where by ``roughly'' we understand that
the process that is behind starts its current phase at most $2\rho$
time after the process that is ahead starts the same phase (at the
same level and round).
We say that $p$ and $q$ are \emph{synchronized at level~$\ell$ and round~$r$} if
$p$ and $q$ are synchronized and $\ell = \ell_p = \ell_q$ and $r = \max(r_p,
r_q)$. Note that at the beginning of the round~$r$ of one of the processes, the
other process might be at round~$r-1$. However, for at least $\Delta'(r)-2\rho$
time, the two processes are at the same round~$r$.
%

Next, we provide a simpler characterization of process synchronization.
Let $t$ be the last time $p$ called $\getNextRoundAndPhase$.  We denote by
$\levelOffset_p = \var{now} - \var{levelStart}$, where $\var{now}$ is the value
returned by~$\fn{now}$ when called by~$p$ at~$t$, and $\var{levelStart}$ is the
sum at line~\ref{line:tsum}. Note that every call to $\getNextRoundAndPhase$ is
preceded by a call to \fn{synchronize}, which in turn calls~\fn{now}.
The next lemma states that we can use level offsets to characterize process
synchronization. We omit its proof, which follows from an analysis of the
\fn{synchronize} and $\getNextRoundAndPhase$ functions.

\begin{lemma}\label{l:condsync}
  After $\tau$, two correct processes $p$ and $q$ are synchronized iff
  $|\levelOffset_p-\levelOffset_q|\leq 2\rho$.
\end{lemma}

\begin{lemma}\label{l:same_blockchain}
  Let $p$ and $q$ be two correct processes. If, after $\tau$, they remain at the
  same level and
  the head of their blockchain has the same
  round, then they are eventually synchronized.
\end{lemma}

\begin{proof}
  Suppose that $p$ and $q$ are both at the same level~$\ell$ and that their
  heads have the same round.
  Furthermore, $p$ and $q$ have already decided at $\ell-1$. From the agreement
  property, $p$ and $q$ agree on the output value at level~$\ell-1$, which means
  that they agree on all blocks up to (and including) level~$\ell-2$, and
  therefore on their rounds as well.
  Thus, the block rounds in $p$'s and $q$'s blockchain are respectively the
  same.

  Next, we observe that both $p$ and $q$ eventually call $\fn{synchronize}$ and
  $\getNextRoundAndPhase$. Indeed, at the end of a round a correct process calls
  \fn{advance}, which in turn calls $\runInstance$ and finally
  \fn{synchronize}. A round eventually terminates, because it has a fixed
  duration. Also, the round returned by \fn{synchronize} will eventually be
  larger than the current round of the process, so the process will eventually
  exit the recursion at line~\ref{line:runCIreccall} and also call
  $\getNextRoundAndPhase$.

  Let $p$ be the first to call $\getNextRoundAndPhase$ and let $t$ be the time
  of the call. Let $t' \geq t$ be the time when $q$ first calls
  $\getNextRoundAndPhase$.

  We first note that $\var{levelStart}$ in the definition of $\levelOffset$ is
  the same for both $p$ and~$q$, at both times $t$ and~$t'$.
  Let $\levelOffset^*_t=t-\var{levelStart}$ and
  $\levelOffset^*_{t'}=t'-\var{levelStart}$. Intuitively, these are the correct
  level offsets at $t$ and $t'$ of any correct process if its local clocks were
  precise.
  We consider the values of the variable $\levelOffset_p$ at~$t$ and~$t'$ and
  denote these by (simply) $\levelOffset_p$ and $\levelOffset'_p$,
  respectively. We note that $\levelOffset'_p - \levelOffset_p = t'-t$, because
  we assume that a process measures intervals of time precisely.
  Given the bound on clock skews, we have that
  $|\levelOffset_p-\levelOffset^*_t|\leq \rho$ and
  $|\levelOffset_q-\levelOffset^*_{t'}|\leq \rho$. Summing up, by using the
  inequality $|a-b|\leq|a|+|b|$, we obtain that $|\levelOffset_q-\levelOffset_p
  - (t'-t)|\leq 2\rho$, that is, $|\levelOffset_q-\levelOffset'_p|\leq
  2\rho$. Then, by Lemma~\ref{l:condsync}, $p$ and $q$ are synchronized at $t'$.
  \qed
\end{proof}

\begin{lemma}\label{l:recv_within_phase}
  If $P$ is a set of correct processes that are synchronized
  after~$\tau$ at a level and a round $r$ with $\Delta'(r)>\delta+2\rho$,
  and a process $p\in P$ sends a message at the beginning of its
  current phase~$\var{ph}$, then this message is received by all
  processes in $P$ by the end of their phase~$\var{ph}$.
\end{lemma}
\begin{proof}
  %
  Assume that $p$ sends its message at time
  $t_p=S_p(\ell,r,\var{ph})$. Consider a process~$q\in P$, and let
  $t_q=S_q(\ell,r,\var{ph})$. Process~$q$ receives the message at most at
  time $t_p+\delta$.
  By the synchronization hypothesis, we have that $t_p-t_q \leq 2\rho$.
  Then we obtain that $t_p+\delta \leq t_q+2\rho+\delta
  < t_q+\Delta'(r)$. Note that $t_q+\Delta'(r)$ is
  the time of the end of the phase~$\var{ph}$ for $q$.
  Note also that if $t_p < t_q$, $q$ might receive the message while
  it is still at round $r-1$. Even so, this message is available
  to $q$ at round $r$ because correct processes keep messages from a
  round one unit higher than their current round.
  \qed
\end{proof}

\begin{lemma}\label{l:allSynchronizedHighestLockProposes}
  Let $\ell$ be a level and $r$ a round with $\Delta'(r)>\delta+2\rho$.
  Consider that all correct bakers are synchronized at
  level~$\ell$ and round~$r$ at a time after~$\tau$.
  Let $p$ be the proposer at round~$r$.
  %
  If $p$ is correct and $\validRound_p\geq \lockedRound_q$ for any correct
  baker~$q$, then all correct bakers decide at level~$\ell$ at the
  end of round~$r$.
\end{lemma}

\begin{proof}
  From Lemma~\ref{l:recv_within_phase}, we obtain that the $\ProposeM$
  message of process~$p$ is received by all correct bakers by the
  beginning of their phase~\PreendorseP.
  Let $\var{eR}$ be the value of the endorsable round field of the \ProposeM message.
  Note that $\var{eR} = \validRound_{p}$.
  We prove next that each correct baker sends the message
  $\PreendorseM(\ell,r,h,u)$, where $u, h$ are the value and the
  predecessor hash proposed by~$p$.
  Let $q$ be some correct baker.
  If $q$ is either unlocked or locked on~$u$, then the condition in
  line~\ref{propCheckLock} holds, and therefore $q$ sends its
  preendorsement for~$(u, h)$.
  Suppose now that $q$ is locked on a value different from~$u$.
  By hypothesis, we have $\lockedRound_{q} \leq \validRound_{p}$, therefore
  $\lockedRound_{q} \leq \var{eR}$.
  %
  %
  %
  Also, we have that $\validRound_p < r$, since $\validRound_p$ is
  set during the execution of $\waitMsg$ before sending the
  $\ProposeM$ message in round~$r$, therefore it is set at a previous
  round.
  %
  %
  We thus have that $\lockedRound_q \leq \var{eR} < r$. If $\lockedRound_q =
  \var{eR}$ then, by quorum intersection, $\lockedBlock_q = u$ thus the first
  disjunction in line~\ref{propCheck1} holds for~$q$. If $\lockedRound_q <
  \var{eR} < r$ then the second disjunction in line~\ref{propCheck1} holds
  for~$q$ (note that $r=r_p=r_q$). Thus $q$ sends the corresponding
  $\PreendorseM$ message.
  So, we have proved that all correct bakers broadcast the
  $\PreendorseM(\ell,r,h,u)$ message (line \ref{propCheck1End}).

  By Lemma~\ref{l:recv_within_phase} again, all these (at least
  $2f+1$) $\PreendorseM(\ell,r,h,u)$ messages are received by all
  correct bakers by the beginning of the phase~\EndorseP. If
  follows that, for all correct bakers, the condition in
  line~\ref{Endorse} is true and thus all correct bakers broadcast
  the $\EndorseM$ message for $(u, h)$ (line~\ref{EndorseEnd}). In
  the next phase (namely the phase~\ProposeP) the quorum condition (line
  \ref{line:quorumCond}) holds for $(u, h)$, for all correct
  bakers, so all of them decide $(u, h)$.
  \qed
\end{proof}

\begin{lemma}\label{l:allSynchronized}
 If at some time after~$\tau$ all correct bakers are synchronized
 at some level~$\ell$ and round~$r$ with $\Delta'(r)>\delta+2\rho$, then all
 correct bakers decide at level~$\ell$ by the end of
 round~$r+f+1$.
\end{lemma}

\begin{proof}
  We first remark that, after $\tau$, thanks to synchrony, a correct baker never
  skips a round, and in particular never skips its turn when it is time to
  propose. Indeed, when a baker calls again \fn{synchronize} and
  $\getNextRoundAndPhase$ to resynchronize at the end of a round, its local
  clock can be in advance with respect to the previous reading of its local
  clock by at most $2\rho$. As we assumed that $2\rho<\Delta'(1)$, the baker
  would still have remaining time to execute the \ProposeP phase of the next
  round.

  Let $p_0, p_1, \dots$ be the sequence of bakers in the order in
  which they propose starting with round~$r$. That is, $p_i$ is the
  proposer at round~$r+i$, for $i\geq 0$.
  Let $j,k$ be the indexes of the first and second correct bakers
  in this sequence.
  As there are at most $f$ Byzantine processes among
  $\set{p_0,\dots,p_k}\setminus \set{p_j}$, we have $j<k\leq f+1$.
  We show next that all correct bakers decide by the end of
  round~$r+k$.

  Suppose first that $p_j$ is such that $\validRound_{p_j}\geq
  \lockedRound_q$, for any correct baker~$q$. By
  Lemma~\ref{l:allSynchronizedHighestLockProposes}, all correct
  bakers decide at the end of round $r+j$.

  Suppose that there is a correct baker with a locked round higher than
  $\validRound_{p_j}$.
  Let $q$ be the baker with the highest locked round among all
  correct bakers.
  In the round at which $p_j$ proposes, that is, in round~$r+j$, $q$
  sends a preendorsement QC that justifies its locked round in the
  \PreendorseP phase (line~\ref{preendBc2}). By Lemma~\ref{l:recv_within_phase},
  this preendorsement QC is received by all correct bakers, who
  update in the \EndorseP phase of round $r+j+1$ their endorsable
  round to $q$'s locked round at line~\ref{line:indirectEndoUpdateRound}.

  If between rounds $r+j+1$ and $r+k-1$ no correct baker updates
  its locked round then the proposer $p_{k}$ will have at round
  $r+k$ that $\validRound_{p_k}\geq\lockedRound_{q}$, for any correct
  baker~$q$. Again, by
  Lemma~\ref{l:allSynchronizedHighestLockProposes}, we conclude that
  at the end of round $r+k$ all correct bakers decide.

  If instead there is a correct baker that updated its locked round
  before round $r+k$, then let $q$ be the baker which updates it
  last, at some round $r+j'$ with $j'<k$.
  When $q$ changes its locked round, $q$ has seen a prendorsement QC
  for round~$r+j'$. This QC is sent together with the $\EndorseM$
  message in the phase~\EndorseP, and therefore it will be received
  by all correct bakers at the beginning of the next phase~\ProposeP.
  Thus every correct baker, including $p_k$, sets its
  $\validRound$ to $r+j'$.
  Because $j'$ is maximal, no correct baker changes its locked
  round between rounds $r+j'+1$ and $r+k-1$. Therefore, at round
  $r+k$, for any baker~$q$, we have that $\lockedRound_{q} \leq
  r+j' = \validRound_{p_k}$. Again, by
  Lemma~\ref{l:allSynchronizedHighestLockProposes} we conclude that at
  the end of round $r+k$ all correct bakers decide.
  \qed
\end{proof}

\begin{theorem}
  \ourproto satisfies progress.
\end{theorem}
\begin{proof}
  We reason by contradiction. Suppose first there is a level~$\ell\geq 1$ such
  that no correct process decides at~$\ell$. Clearly, $\ell$ is minimal with
  this property.

  We first show that eventually all correct processes are synchronized.
  As $\ell$ is minimal, we have that there is at least one
  correct process that has decided at~$\ell-1$.

  As processes invoke $\pullChain$ at regular intervals, all
  correct process will eventually be at level $\ell$ (that is, they
  will have decided at $\ell-1$).  %
  We show next that all correct processes have the same blockchain
  head.
  Let $p$ be a correct process that has its $\headQC_p$ for the block
  with the lowest round at level~$\ell-1$.  Process $p$ will
  eventually receive a $\pullChain$ request at some point after $\tau$
  and it will answer. If each correct process~$q$ has
  $\validRound_q=0$ at the time of the receipt of $p$'s answer, then
  every correct process accepts $p$'s branch, by the definition of
  $\betterChain$.
  Suppose however that there is a process $q$ that has
  $\validRound_q>0$ when it receives $p$'s answer. In this case
  consider a time when round durations are so big that $I$ and
  $\Delta$ are very small in comparison. More precisely, there is a
  time period when all $\pullChain$ requests and their answers happen during a
  period when correct processes update their states only in response
  to a $\NewChain$ event, but not in response to $\NewMessage$
  events. Such a period exists because regular messages are sent only
  at phase boundaries. This means that the chain ending with the
  proposal with the highest endorsable round~$r$ will be seen by all
  correct processes, and these processes will have their endorsable
  round smaller or equal to~$r$. They will update their blockchains to
  this chain (if they were on a different one). Note that if two
  processes have the same endorsable round then they also have the
  same blockchain; this can be seen by a simple application of the
  quorum intersection property.
  We have this obtained that eventually all correct processes have
  the same blockchain (head).
  We can therefore apply Lemma~\ref{l:same_blockchain} to obtain
  that there is a time after $\tau$ at which all correct processes are
  synchronized.

  Now, recall that the function $\Delta'$ has the property that there is a round $r$
  such that $\Delta'(r)>\delta+2\rho$. As $\Delta'$ is increasing, this property
  holds for all subsequent rounds as well. And, given that all processes are
  synchronized from some time on, as proved in the previous paragraph, we obtain
  that the hypothesis of Lemma~\ref{l:allSynchronized} is satisfied. Therefore
  all correct processes decide at~$\ell$, which contradicts the assumption that
  no correct process decides at~$\ell$. In other words, we have proved that, for
  any level~$\ell$, there is at least one correct process that decides
  at~$\ell$.

  Finally, we show that for any level~$\ell$, any correct process eventually
  decides at~$\ell$. Suppose that there is a correct process~$p$ that does
  not decide at some level~$\ell\geq 1$. From the first part of the proof we
  obtain that there is at least one other correct process~$q$ that eventually
  decides at $\ell$.
  Process $q$ will eventually receive $p$'s pull request, will
  reply, and $p$ will therefore receive an endorsement QC for
  level~$\ell$ which enables it to decide at~$\ell$. This contradicts
  the assumption, and allows us to conclude.
  \qed
\end{proof}

\end{document}